\documentclass[aps,pra,singlecolumn,superscriptaddress,floatfix,nofootinbib,showpacs,longbibliography,notitlepage,24 pt]{revtex4-1}

\usepackage[utf8]{inputenc}  
\usepackage[T1]{fontenc}     
\usepackage[british]{babel}  
\usepackage[sc,osf]{mathpazo}
\usepackage[scaled=0.86]{berasans}  
\usepackage[colorlinks=true, citecolor=blue, urlcolor=blue]{hyperref}
\makeatletter
\newcommand{\setword}[2]{%
  \phantomsection
  #1\def\@currentlabel{\unexpanded{#1}}\label{#2}%
}
\makeatother
\usepackage{graphicx} 
\usepackage[babel]{microtype}  
\usepackage{amsmath,amssymb,amsthm, bm,amsfonts,mathrsfs,bbm} 

\usepackage{xspace}  
\usepackage{pgf,tikz}
\usepackage{xcolor}
\usepackage{multirow}
\usepackage{array}
\usepackage{bigstrut}
\usepackage{braket}
\usepackage{color}
\usepackage{natbib}
\usepackage{multirow}
\usepackage{mathtools}
\usepackage{float}
\usepackage[caption = false]{subfig}
\usepackage{xcolor,colortbl}
\usepackage{color}
\usepackage{geometry}
\newcommand{\Tr}{\operatorname{Tr}}

\newcommand{\be}{\begin{equation}}
\newcommand{\ee}{\end{equation}}
\newcommand{\ba}{\begin{eqnarray}}
\newcommand{\ea}{\end{eqnarray}}
\newcommand{\ketbra}[2]{|#1\rangle \langle #2|}

\newtheorem{lemma}{Lemma}

\def\>{\rangle}
\def\<{\langle}

\newcommand{\map}[1]{\mathcal{#1}}

\newtheorem{theo}{Theorem}

\newtheorem{cor}{Corollary}
\newtheorem{defi}{Definition}

\usepackage{centernot}
\usepackage{subfig}

\DeclareUnicodeCharacter{202A}{} 
\DeclareUnicodeCharacter{202C}{} 

\definecolor{lime}{HTML}{A6CE39}
\DeclareRobustCommand{\orcidicon}{
	\begin{tikzpicture}
	\draw[lime, fill=lime] (0,0) 
	circle [radius=0.16] 
	node[white] {{\fontfamily{qag}\selectfont \tiny ID}};
	\draw[white, fill=white] (-0.0625,0.095) 
	circle [radius=0.007];
	\end{tikzpicture}
	\hspace{-2mm}
}

\foreach \x in {A, ..., Z}{\expandafter\xdef\csname orcid\x\endcsname{\noexpand\href{https://orcid.org/\csname orcidauthor\x\endcsname}
			{\noexpand\orcidicon}}
}
\begin{document}

\title{Randomness-free Detection of Non-projective Measurements: Qubits \& Beyond}

\author{Sumit Rout\orcidS{}}
\email{sumit.rout@phdstud.ug.edu.pl}
\affiliation{International Centre for Theory of Quantum Technologies (ICTQT), University of Gda{\'n}sk, Jana Ba{\.z}ynskiego 8, 80-309 Gda{\'n}sk, Poland}

\author{‪Some Sankar Bhattacharya‬\orcidB{}}
\affiliation{Física Teòrica: Informació i Fenòmens Quàntics, Universitat Autònoma de Barcelona,
08193 Bellaterra, Spain}
\author{Pawe{\l} Horodecki\orcidP{}}
\affiliation{International Centre for Theory of Quantum Technologies (ICTQT), University of Gda{\'n}sk, Jana Ba{\.z}ynskiego 8, 80-309 Gda{\'n}sk, Poland}

\begin{abstract}
  Non-projective measurements play a crucial role in various information-processing protocols. In this work, we propose an operational task to identify measurements that are neither projective nor classical post-processing of data obtained from projective measurements. Our setup involves space-like separated parties with access to a shared state with bounded local dimensions. Specifically, in the case of qubits, we focus on a bipartite scenario with different sets of target correlations. While some of these correlations can be obtained through non-projective measurements on a shared two-qubit state, it is impossible to generate these correlations using {\it projective simulable} measurements on bipartite qubit states, or equivalently, by using one bit of shared randomness and local post-processing. For certain target correlations, we show that detecting qubit non-projective measurements is robust under {\it arbitrary} depolarising noise, except in the limiting case. We extend this task for qutrits and demonstrate that some correlations achievable via local non-projective measurements cannot be reproduced by both parties performing the same qutrit {\it projective simulable} measurements on their pre-shared state. We provide numerical evidence for the robustness of this scheme under {\it arbitrary} depolarising noise. For a more generic consideration (bipartite and tripartite scenario), we provide numerical evidence for a projective-simulable bound on the reward function for our task. We also show a violation of this bound by using qutrit POVMs. From a foundational perspective, we extend the notion of non-projective measurements to general probabilistic theories (GPTs) and use a randomness-free test to demonstrate that a class of GPTs, called {\it square-bits or box-world} are unphysical.
\end{abstract}

\maketitle
\section{Introduction}
A measuring device is associated with an input physical system and an output classical variable. Any physical theory in addition to specifying the state also prescribes allowed observables and a rule for predicting the outcomes of a measurement. For a two-level quantum system also known as a qubit, in contrast to classical two-level systems (bits), the quantum formalism allows for measuring devices with more than two irreducible outcomes. In other words, there are measurements with more than two outcomes that cannot be realised as a post-processing of the outcomes obtained from a two-outcome quantum measurement. Such measurements are called {\it non-projective measurements}. The usefulness of general measurement devices has been recognised since the late 1980s. Studies have demonstrated their advantages in quantum state discrimination \cite{Ivanovic87, Peres88, DiMario22, QuantumStateDisriminationRev}, entanglement detection \cite{Shang2018}, quantum tomography \cite{Derka1998, Renes2004, OptTomography, singleSettingTom, povmSHADOWGhune}, and quantum metrology \cite{Ragy2016, Szcyykulska2016, rafalMETRO}. More recently, non-projective measurements have been shown to outperform projective ones in information-processing tasks such as randomness certification \cite{Acin16}, quantum cryptography \cite{optimalCRYPTO}, and port-based teleportation \cite{Ishizaka2008, Studzinski2017port, Mozrzymas2018optimal}. Additionally, research has established the advantages of non-projective measurements in quantum computing \cite{algREVIEW, Bacon2006, HSP} and the foundations of quantum information \cite{Barret2002, Vertesi2010}. Consequently, detecting and characterising such non-projective measurements is of practical importance.

In the simplest scenario, one can consider tomography \cite{Lundeen09} of a fixed measurement device by randomly preparing the system to be measured in different states. However, this procedure is resource-intensive and a better protocol can be developed. For example, one may randomise over a smaller set of measurements, which is a powerful tool for probing the non‐classical nature of physical processes. This approach has led to several non-classicality certification tasks with various degrees of knowledge about the process/device \cite{Brunner14, Sekatski23}. In some cases, even a minimal amount of seed randomness turns out to be sufficient for success in non-classicality detection \cite{Barrett11, Putz14}. In 2010 Vertesi et al. \cite{Vertesi10} and later in 2016 Gomez et al. \cite{Gomez16} proposed schemes to detect non-projective measurements in a setting involving two space-like separated parties. They perform local measurements, chosen randomly from a set, on a pre-shared bipartite state. In the prepare and measure setup where Alice and Bob receive random inputs, Tavakoli et al. \cite{Tavakoli20} showed robust semi-device independent self-testing of extremal qubit three- and four-outcomes POVM. Also, Martínez et al. \cite{Martinez23} showed robust self-testing of seven-outcomes ququad non-projective measurement using a similar setup. Both of these works are semi-device independent since they assume the underlying theory describing the communicated system to be quantum with an upper bound on the associated Hilbert space dimension.

But the fact remains that classical devices can only generate pseudo-random sequences (which might be sufficient for certain practical purposes), making it in principle necessary to use non-classical devices to certify the non-classicality of other processes. In an attempt to overcome this circular nature of the certification technique, some proposals have been forwarded. In \cite{Renou19} the authors consider a network scenario with multiple non-classical sources and fixed measurements are performed locally. Although this approach does not require any knowledge of the internal workings of the devices, it assumes that the sources are independent. In another approach \cite{Guha2021, Ma2023}, the non-classicality of a fixed source as well as a measuring device has been detected, while an upper bound on the operational dimension of the system under consideration requires to be known. The present work follows the latter approach and unveils a larger set of correlation simulation tasks that can be used for detecting non-projective measurements.

Specifically, in \cite{Ma2023} a proof of concept realization with photonic systems for the certification of three-outcome qubit non-projective measurements without seed randomness has been demonstrated. The scheme requires a minimal assumption that the dimension of parts of the system is a priori known to the experimenter. The authors show that qubit non-projective measurements are necessary for producing specific correlated (public) coins from pairs of photons entangled in their transverse spatial modes. The possibility of a similar scheme for higher dimensional measurements was left open. 

In the present article, we answer several open questions mentioned in \cite{Ma2023} regarding randomness-free detection of non-projective measurements. First, we extend the notion of non-projective measurements to general probabilistic theories (GPTs) \cite{Barrett11,Janotta11,Scandolo21} and propose an analogous problem of their detection. Second, we expand the range of target correlations used for detecting qubit non-projective measurements and also the noise-robustness of the proposed schemes. Last, we explore the detection of higher dimensional quantum non-projective measurement with different numbers of outcomes.

We consider the problem of certifying whether some uncharacterised measurement device is a genuine measurement with a fundamentally irreducible number of outcomes or otherwise. We only assume the knowledge regarding the input dimension of the measuring device. Given an arbitrary measuring device, such a characterization requires the use of other known measurements and randomness to perform prepare and measure experiments \cite{Tavakoli20, Martinez23}. Instead in this work, we consider a different setting where multiple copies of the same uncharacterised measurement device are available to spatially separated multiple labs. These fixed measurements are performed on a multipartite system shared between the labs. This has a two-fold advantage in comparison to the previous approach: (i) it cuts the cost of having seed randomness usually considered a costly resource \cite{Ambainis08} and (ii) it does not necessitate the use of additional fully characterised measurement devices. Thus, our proposed framework is characterised by the tuple $(n,d,k)$, where $n$ is the available number of copies of the uncharacterised measuring devices which are spatially separated, $d$ is their input dimension, and $k(>d)$ is the number of outcomes of each measuring device.
\subsection{Outline}

In section \ref{sec:Prelimianary}, we introduce the notion of projective simulable measurements as well as non-projective measurements in quantum theory and its extension to GPT. Then we propose an operational task and two different criteria using these tasks, which we would consider for detecting qudit non-projective measurements. Next, in section \ref{sec:qubit non proj}, we use these tasks to detect qubit non-projective measurement. First in sub-section \ref{Subsec:QPS bound} we show that all the correlations obtained using classically correlated systems and post-processing of outcomes can also be obtained using projective simulable measurements on quantum systems with equal local dimension and vice versa. Then in sub-sections \ref{sec:S223} and \ref{sec:S224} we discuss the detection of three and four-outcome qubit non-projective measurements. We also show the robustness against arbitrary noise of the proposed detection schemes. For these detection schemes, we do not use any assumption regarding the projective simulable measurements performed by each of the parties, {\it i.e.} they could be different in general. However, for the result in sub-sections \ref{sec:S223S} and \ref{sec:S235S}, we assume that even the projective simulable measurement devices for the parties are identical copies. In sub-section \ref{sec:S223S}, we demonstrate the impossibility of obtaining some correlations while using two copies of shared qubit states adaptively and performing identical qubit projective simulable measurements. However, these target correlations can be achieved using non-projective qubit measurements. In section \ref{sec:qutrit non-proj}, we discuss the tasks to detect qutrit non-projective measurement. In sub-sections \ref{sec:S235} and \ref{sec:S335}, we discuss the detection of five-outcome qutrit non-projective measurements involving two and three non-communicating parties respectively. We provide numerically obtained bounds on the figure of merit using qutrit projective simulable measurements and a violation of it using qutrit POVM. In sub-section \ref{sec:S235S}, we present analytical proof for the detection of five-outcome qutrit non-projective measurements in a bipartite scenario while assuming that the parties use identical measuring devices. We also provide evidence for the robustness of the detection scheme against arbitrary noise in this case. In section \ref{sec:gpt non-proj}, we provide a task for a randomness-free test to show that some special kinds of GPT called boxworld are unphysical. Lastly in the section \ref{Discussion}, we conclude with a summary of our results and some open directions.

\section{Preliminaries}\label{sec:Prelimianary}
This section will briefly introduce the necessary preliminaries and notations used throughout the article. 
  
\subsection{Quantum Measurements}
The state space associated with a quantum system is given by the set of positive semi-definite operators with a unit trace $D(\mathbb{C}^d)$ defined over some associated complex Hilbert space $\mathbb{C}^d$. The measurement in any physical theory is a mapping from the state space associated with a system to real numbers. In quantum mechanics, a general measurement is given by a set of positive semi-definite operators over Hilbert space $\mathbb{C}^d$, {\it i.e.},  $\{E_i\}_{i=0}^{n-1}$ that sum up to identity, {\it i.e.}, $\sum_{i=0}^{n-1} E_i=\mathbb{I}$. Such a measurement is called a Positive Operator Valued Measure (POVM). A projective measurement (PVM) is a specific type for which any two of the positive operators $\{F_i\}_{i=0}^{d-1}$ are orthogonal projectors    $F_iF_j=F_i\delta_{ij}$. The projectors are extreme points in the convex set formed by the elements of POVM in quantum theory. One can also define a class of POVMs on a $d$-dimensional quantum system with $n$ outcomes ($n > d$), which can be realised by post-processing the outcomes of a $d$-outcome projective measurement performed on the same system. 
\begin{defi}
A qudit $k$-outcome POVM $\mathcal{E}:=$ $\{E_i\}_{i=0}^{k-1}$ is projective-simulable if $E_i=\sum_{j} P_{ij} F_{j}$ where $\{P_{ij}\}_{i=0}^{k-1}$ is a valid probability distribution $\forall j\in\{0,\cdots,d-1\}$ and $\{F_{j}=|\psi_j\rangle\langle\psi_j|, \ket{\psi_j}\in\mathbb{C}^d\}_{j=0}^{d-1}$ is a qudit projective measurement {\it i.e.} $F_{j}F_{k}=F_{j}\delta_{j,k}$. 
\end{defi}
However, there exist POVMs with $n$ outcomes ($n > d$) that can be performed on a $d$-dimensional quantum system but cannot be simulated by post-processing the outcomes of any projective measurement on the same system. Throughout this work, we will refer to such qudit measurement that is not projective-simulable as {\it non-projective simulable} (or non-projective in short). These measurements can be realised as projective measurements which are performed on a quantum system with an extended Hilbert space \cite{Nielsen2012}. 
In other words, exactly simulating a non-projective measurement in $d$ dimensional Hilbert space using some projective measurement requires some additional ancilla along with the $d$ dimensional quantum system. Recently, it was shown that the simulation of most general non-projective measurements on $m$ qubits using projective measurements requires an additional $m$-qubit ancilla \cite{Oszmaniec2017}. Note that there is another concept of simulation of non-projective measurement, namely probabilistic simulation where one is allowed to reject some of the trials. The probabilistic simulation of any non-projective measurement using a projective measurement requires a single qubit ancilla \cite{Singal2022}.

As an example of the situation when a given POVM is non-projective simulable, consider the $3$-outcome qubit POVM $\{E_i=\lambda_i(\mathbb{I}+\hat{n}_i\cdot \vec{\sigma})\}_{i=0}^2$ where $\lambda_i>0$, $\vec{\sigma}=(\sigma_1,\sigma_2,\sigma_3)$ is a vector of Pauli matrices, $\sum_i \lambda_i=1$, $\hat{n}_i\in \mathbb{R}^3$, $|\hat{n}_i|^2=1$ (unit vector), $\sum_i \lambda_i \hat{n}_i=0$ and $\hat{n}_i\neq \hat{n}_j~\forall i,j$. These qubit POVMs are not simulable using qubit projective measurements. Another such example is four outcome qubit SIC-POVM $\{\Pi_{\alpha}=\frac{1}{2}\ket{\psi_{\alpha}}\bra{\psi_{\alpha}}:\ket{\psi_{0}}=\ket{0},\ket{\psi_{\alpha}}=\sqrt{\frac{1}{3}}\ket{0}+  \sqrt{\frac{2}{3}}e^{\frac{2\pi i}{3}(\alpha-1)}\ket{1} \text{ for } \alpha\in\{1,2,3\} \}$. Here we are interested in detecting such measurements using correlations obtained under some operational constraints. With this aim, we will later define operational tasks to this end. For these tasks, we will additionally show that the statistics obtained from projective simulable measurements can also be reproduced using classical systems with equivalent local dimensions. In this sense, projective simulable measurements correspond to a notion of classicality for measurements. Although we assume that the systems are completely described by quantum theory and have some bounded dimension, such tasks might not require the entire internal description of the devices to be restricted to quantum theory, {\it i.e.}, the possibility of achieving these tasks perfectly is sensitive to a notion of non-classicality that is theory-independent in general. We consider the framework of General Probabilistic Theory (GPT) to introduce this notion of non-classicality for measurements.

\subsection{Measurements in GPTs}
 A General Probabilistic Theory $\mathcal{X}$ is specified by a list of system types and the composition rules specifying a combination of several systems. A system $S$ is described by a state $\omega$ which specifies outcome probabilities for all measurements that can be performed on it. For a given system, the set of possible normalised states forms a compact and convex set $\Omega$ embedded in a positive convex cone $V_+$ of some real vector space $V$. Convexity of state space $\Omega$ assures that any statistical mixture of states is also a valid state. The extremal points of $\Omega$, that do not allow any convex decomposition in terms of other states, are called pure states or states of maximal knowledge. An effect $e$ is a linear functional on $\Omega$ that maps each state onto a probability, \textit{i.e.}, $e:\Omega\mapsto[0,1]$ by a pre-defined rule $p(e|\omega)=Tr(e^T.\omega)$. The set of effects $\mathcal{E}$ is embedded in the positive dual cone $(V^\star)_+$. The normalization of $\omega$ is determined by $u$ which is defined as the unit effect and a specified element of $(V^\star)_+$, such that, $p(u|\omega)=Tr(u^T.\omega)=1, \forall \omega\in\Omega$. Assuming {\it no-restriction hypothesis} \cite{Chiribella11}, a $k$-outcome measurement is specified by a collection of $k$ effects, $M\equiv\{e_j~|~\sum_{j=1}^ke_j=u\}$, such that, $\sum_{j=1}^kp(e_j|\omega) = 1$, for all valid states $\omega$. A measurement is called {\it sharp} if all its elementary effects $\{e_j\}_j$ correspond to the extreme points of the set of effects. Another much-needed component to complete the mathematical structure for GPT is the reversible transformation $\mathcal{T}$ which maps states to states, {\it i.e.}, $\mathcal{T}(\omega)= \omega'\in\Omega$. They are linear to preserve the statistical mixtures, and they cannot increase the total probability. In a GPT one can introduce the idea of distinguishable states from an operational perspective which consequently leads to the concept of \textit{Operational dimension}.

\begin{defi}
Operational dimension of a system is the largest cardinality of the subset of states, $\{\omega_i\}_{i=1}^n\subset\Omega$, that can be perfectly distinguished by a single measurement, {\it i.e.}, there exists a measurement, $M\equiv\{e_j~|~\sum_{j=1}^ne_j=u\}$, such that, $p(e_j|\omega_i)=\delta_{ij}$.
\end{defi}

The classical and quantum theory are two examples of GPTs. In each of these theories, a system with operational dimension two corresponds to a bit and a qubit respectively. Note that the state space of a qubit can be embedded in $\mathbb{R}^3$. Also, the operational dimension of a system may not necessarily be equal to the dimension of vector space in which the state space resides \cite{Brunner14,Dallarno17}. Later we also discuss (see section \ref{sec:gpt non-proj}) a class of hypothetical theories with operational dimension $2$.

Within the framework of GPT, a generalised notion of non-projective simulable measurements as in quantum theory is given by the non-sharp simulable measurements. These measurements are complementary to the set of sharp simluable measurements which are defined as follows.

\begin{defi}
A $k$-outcome measurement $\mathcal{E}:=$ $\{e_i\}_{i=1}^k$ is sharp-simulable for a GPT system $S$ of operational dimension $d$ if $e_i=\sum_{j} P_{ij}\pi_{j}$ where $\{ P_{ij}\}_{i=1}^{k}$ is a probability distribution for all $j\in \{1,\cdots,d\}$ and $\{\pi_{j}\}_{j=1}^{d}$ is a sharp measurement for the GPT system $S$.
\end{defi}

Clearly, the non-projective simulable measurements are a special case of not sharp simulable measurements for the quantum theory. As discussed earlier, non-projective measurements are valuable resources for many information-processing tasks. Thus, their detection is an important problem. There can be various settings where one can attempt to detect such measurements. The first hurdle one faces while detecting non-projective measurements in a prepare and measure setup is that it cannot be done by using a single quantum system. This is because any input-output correlation generated by a $d$-level quantum system and a fixed measurement device can also be simulated by a $d$-level classical system \cite{Frenkel15}. Thus for bipartite systems, the task is to find alternate setups and generate correlations which cannot be explained in terms of $d$-outcome measurements performed on the sub-systems. 

A possible approach is to consider the Bell nonlocality scenario, where independent measurements are performed on a spatially separated bipartite system and only the conditional input-output probability is analysed. There exist Bell inequalities where three-outcome qubit measurements can outperform fundamentally two-outcome measurements \cite{Kleinmann17}. An essential assumption in Bell's nonlocality scenario is freedom of choice. Operationally it means that at each spatially separated lab, some seed randomness is required while performing one of many measurements. This situation could be overcome with a setting where only a fixed measurement is performed on the subsystems. To this end, we will use the Local Operation (LO) and {\it bounded local dimension} framework where space-like separated parties can perform local operations for free and are allowed to pre-share systems of bounded local dimension. Here, we consider the pre-sharing of additional resources to be costly. In the following, we describe a setting where we can detect non-projective measurements without assuming
freedom of choice in a scenario where parties are space-like separated.

\subsection{The Task $\mathbb{G}~(n,d,k)$}

The setup for the correlation simulation task we now define, in its most general form, involves $n$ spatially separated parties $A_1$, $A_2$, $\cdots$, $A_n$ (see figure \ref{fig:Setupndk}). Let ${\map P}_{A_1,A_2,\cdots,A_n}$ be a preparation device that prepares an $n$-partite state with local operational dimension $d$ and distributes them among these parties. The parties do not have access to any additional shared correlations however they have access to local sources of randomness. Additionally, each of them has access to an uncharacterised measuring device ${\map M}_{A_i}$ having $k$-outcomes $(k>d)$. In some special cases, the measuring devices would be assumed to be identical copies. After performing measurements on their respective subsystems, the $n$ parties output  $a_1$, $a_2$, $\cdots$, $a_n\in\{0,1,\cdots,k-1\}$, respectively. The joint probability of these outcomes is denoted as $p(a_1, a_2, \cdots, a_n)$ where $a_i$ is the output of $A_i$. Note that the devices in general could be described by a theory $\mathcal{X}:=(\Omega_d,\mathcal{E}_d,\mathcal{P})$ with operational dimension $d$ where $\Omega_d,\mathcal{E}_d,\mathcal{P}$ are the state space, effect space and the probability rule respectively. Here, we will use the notation $\mathcal{C}(k):=\{P:P=\{p(a_1,\cdots,a_n)\}_{a_1,\cdots,a_n=0}^{k-1}\}$ for the set of all valid joint probability distribution. If the parties share a multipartite correlated state and perform some local measurement allowed by the theory $\mathcal{X}$ then $\mathcal{C}^{\mathcal{X}}_d(k)\subset \mathcal{C}(k)$ will denote the subset of correlation achievable by the theory $\mathcal{X}$ while using a system of local operational dimension $d$. For classical and quantum theory $\mathcal{X}=Cl$ and $\mathcal{X}=Q$ respectively. When the local measurements are restricted to be projective simulable in quantum theory then we shall refer to this set of correlations as $\mathcal{C}^{PQ}_d (k)$. Clearly, $\mathcal{C}^{nPS}_d (k):= \mathcal{C}^{Q}_d (k)~ \backslash ~\mathcal{C}^{PQ}_d (k)$ denotes the quantum correlations that can be achieved by performing some non-projective simulable measurement ($nPS$) on shared state with local dimension $d$.

In this scenario, we shall use two different means to detect resources of our interest, {\it i.e.}, non-projective measurements. Firstly, we can have a convex set of target correlations, say $T[n,d,k]$, that we would like to simulate. In such a case, we define the winning condition as obtaining any correlation in this set. A typical question in such scenarios is regarding robustness against noise. In this direction, for some explicit tasks, we will later show that qudit non-projective measurements yield some correlations in $T[n,d,k]$ even in the presence of noise. Secondly, we can define a quantity, ${\map R}[\mathbb{G}~(n,d,k)]$ as a figure of merit, which is a function (not necessarily linear) of the joint probability distribution $\{p(a_1,\cdots,a_n)\}_{a_1,\cdots,a_n}$. If this quantity ${\map R} [\mathbb{G}~(n,d,k)]$ has a threshold value for measuring devices with fundamentally $d$-outcomes, obtaining some value beyond this threshold will detect non-projective, {\it i.e.} fundamentally more than $d$-outcome measurements. Using an operational theory, the objective would be to achieve a target correlation using a shared state and effect in the first case while the second case requires one to optimise the quantity ${\map R}[\mathbb{G}~(n,d,k)]$ over all states and effects given an upper bound on the operational dimension of the local sub-systems. We will specify a priori for each task whether we will be using the former or latter for detecting non-projective measurements. Additionally, if the system, measurement, and probability rule are given by some theory $\mathcal{X}$ then the class target correlations or the figure of merit could be used to prove that the theory is not physical. Specifically, later we will use one such task to show that certain polygon GPTs are unphysical. We would refer to the operational task defined above in short as $\mathbb{G}~(n,d,k)$ henceforth. In case the uncharacterised measuring device of each party is assumed to be identical then we will use $\mathbb{G}^{sym}~(n,d,k)$ to denote the task.

\begin{figure}[h!]
    \centering
    \includegraphics[scale=0.1]{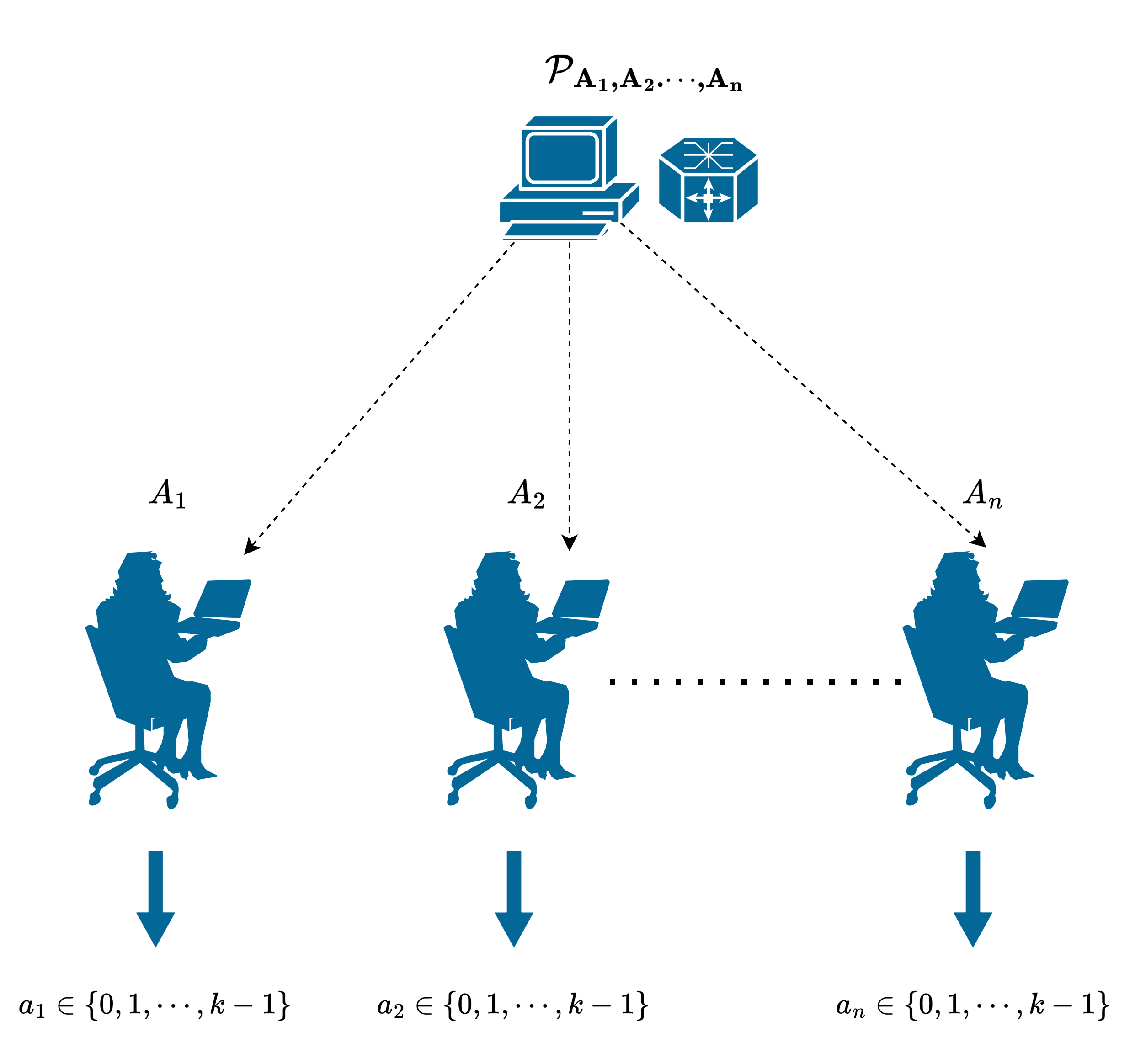}
    \caption{The task $\mathbb{G}~(n,d,k)$ considers $n$ spatially separated parties $A_1$, $A_2$, $\cdots$, $A_n$ who only pre-share a multipartite system, each with local {\it operational dimension} $d$, prepared using ${\map P}_{A_1,A_2,\cdots,A_n}$ and have access to some uncharacterised during device. Using these resources they generate outcomes $a_1$, $a_2$, $\cdots$, $a_n\in\{0,1,\cdots,k-1\}$ respectively. The joint probability of these outcomes is later used to detect qudit non-projective measurements using a set of target correlations or a payoff function.}
    \label{fig:Setupndk}
\end{figure}

\subsubsection{Merit of the task}
Let us consider a case where we define payoff function ${\map R}[\mathbb{G}~(n,d,k)]$ for the task $\mathbb{G}~(n,d,k)$. The parties can obtain a correlation $p(a_1, a_2, \dots, a_n)$ using a multipartite correlated state $\omega_{A_1, A_2,\dots A_n}\in \Omega_d$ with local operational dimension $d$ and a measuring device $\{e^{A_i}_j\}_{j=1}^k\in \mathcal{E}^{A_i}_d$ for $i^{th}$ party. Let the probability rule $\mathcal{P}$ is prescribed by a theory $\mathcal{X}:=(\Omega_d,\mathcal{E}_d,\mathcal{P})$. Then the objective of the parties in the theory $\mathcal{X}$ is to optimise the merit of the task:

\begin{equation}
{\map R}_{max}^{\mathcal{X}}[\mathbb{G}~(n,d,k)]=\max_{\Omega_d,\mathcal{E}_d}{\map R}^{\mathcal{X}}[\mathbb{G}~(n,d,k)]
\end{equation}

where ${\map R}^{\mathcal{X}}[\mathbb{G}~(n,d,k)]$ denotes the payoff obtained using a shared state $\omega_{A_1, A_2,\dots A_n}$ and local effects $\{e^{A_i}_j\}_{j=1}^k$ for the $i^{th}$ party.

\subsubsection{Classical $d$-Bound}
When using a figure of merit ${\map R}[\mathbb{G}~(n,d,k)]$ to detect non-projective measurements we would like to obtain an upper bound on the payoff while using only qudit projective simulable measurements on a pre-shared quantum state. It would involve optimisation over all the shared states with local dimension $d$ and projective simulable measurements. However, consider the following result for this task which we will formally state later (see theorem \ref{theo:classicalcorr}).  The set of quantum correlations generated from projective simulable measurements on a shared quantum system of local dimension $d$ is equal to the correlations generated by sharing a classical system with the same local dimension. Thus, we can equivalently optimise the payoff function over the set of probability distributions $\{P = \{p(a_1, \dots, a_n)\}_{a_1, \dots, a_n = 0}^{k-1}\}$ 
derived from local post-processing of joint distributions where each party has \(d\)-outcomes. A violation of the obtained bound using some quantum 
measurement implies that the correlation generated is outside the classical set. As a result, the
payoff functions help witness the presence of non-projective measurements.  Note that the set of all the joint probability distributions with higher than $d$-outcome obtained using projective simulable measurements is non-convex and therefore this approach would be useful for us.

The state space $\Omega_d$ of the classical systems with local measurement dimension $d$ is a $d^n$-simplex $\mathbb{R}^{d^n}$ and any state $\omega_{A_1,\cdots, A_n}^{Cl}\in\mathbb{R}^{d^n}$ can be defined as $\omega_{A_1,\cdots, A_n}^{Cl}=(\lambda_{i_1,i_2,\cdots,i_n})_{i_1,i_2,\cdots,i_n\in\{0,1,\cdots d-1\}}$ such that $\sum_{i_1,i_2,\cdots,i_n=0}^{d-1} \lambda_{i_1,i_2,\cdots,i_n}=1$ and $\lambda_{i_1,i_2,\cdots,i_n}\ge 0$. For the task $\mathbb{G}~(n,d,k)$, the $n$ parties after obtaining an outcome can each locally apply some stochastic map $S_{d\rightarrow k}^{A_i}:=\left\{s_{l,m}\right\}_{l=1,\dots,k \atop m=1,\dots,d}$ where $s_{l,m}\geq 0~\forall ~l,m$ and $\sum_{l} s_{l,m}=1~\forall~m$.

Now, the projective simulable bound (${\map R}_{max}^{PQ}[\mathbb{G}~(n,d,k)]$) which is also same as classical bound (${\map R}_{max}^{Cl}[\mathbb{G}~(n,d,k)]$) on the figure of merit for the task is:
\begin{equation}
{\map R}_{max}^{PQ}[\mathbb{G}~(n,d,k)]=\max_{\mathbb{R}^{d^n},~\{S_{d\rightarrow k}^{A_i}\}_{i=1}^n}{\map R}^{Cl}[\mathbb{G}~(n,d,k)]
\end{equation}

When the task involves two parties only and the local dimension of the shared classical system is restricted to be two then the above optimisation problem becomes easier to solve for some figure of merits. This is a consequence of the result that we have stated below. It will be used to characterise correlations with more than two outcomes arising from sharing a bipartite classical state with local operational dimension two. 

\begin{lemma}\label{lemma:generator2}(Guha {\it et al.} \cite{Guha2021})
    All the bipartite classically correlated state of local operational dimension two, {\it i.e.} $\omega_{A, B}^{Cl}=(\lambda_{i,j})_{i,j\in\{0,1\}}$ such that $\lambda_{i,j}\geq 0$ and $\sum_{i,j=0}^1\lambda_{i,j}=1$, can be obtained from a two-level classically correlated state $(\lambda,0,0,1-\lambda)$ where $\lambda\in[0,1]$ and applying local stochastic maps. 
\end{lemma}

Thus, while optimizing over two-level classically correlated states, without loss of generality, we can assume that the initial shared state is $(\lambda_{0,0}=\lambda,\lambda_{0,1}=0,\lambda_{1,0}=0,\lambda_{1,1}=1-\lambda)$ where $\lambda\in[0,1]$. 

\section{Detecting Qubit Non-projective Measurements}\label{sec:qubit non proj}

In this section, we will first show the equivalence between sharing a classical and quantum system of equal local dimension when only projective simulable measurements are allowed in the task $\mathbb{G}~(n,d,k)$ introduced in the previous section. This will be useful in obtaining the projective simulable bound for the payoff (Sec.\ref{Subsec:QPS bound}) and showing the impossibility of obtaining any correlation in the set of target correlations using projective simulable measurements. Then we will provide schemes for the detection of three (see Sec.\ref{sec:S223}) and four-outcome (see Sec.\ref{sec:S224}) qubit non-projective measurements using a class of target correlations. In each of these subsections, we will first prove the impossibility of producing any of the target correlations using projective simulable qubit measurement.  Subsequently, we will show explicit quantum state and non-projective measurements that yield a subset of target correlations. Furthermore, we will also show the robustness against noise for the detection of qubit non-projective measurements in each of the subsections. Finally in Sec.\ref{sec:S223S}, we will consider a situation where local measurement devices are identical. We show that by sharing even two bipartite systems, each of local operational dimension two, and using them adaptively, it is impossible to obtain certain correlations that arise from a qubit non-projective measurement.

\subsection{Quantum Projective-Simulable $d$-Bound}\label{Subsec:QPS bound}
In the following, we show that the correlations obtained from quantum projective-simulable measurements on systems with local operational dimension $d$ are the same as the correlations that can be obtained from sharing classical systems with operational dimension $d$ for all tasks that we have considered here.

\begin{theo}\label{theo:classicalcorr}
For task $\mathbb{G}~(n,d,k)$, $\mathcal{C}^{Cl}_d (k) = \mathcal{C}^{PQ}_d (k)$. 
\end{theo}

{\it Outline of proof:} This follows from the following observation. For a local system of operational dimension $d$, joint outcome probabilities from any quantum state and projective measurement are diagonal elements of the density matrix for the state in the measurement basis. Thus, the same statistics could also be obtained from a classically correlated (diagonal) state and measurement in the computational basis. The converse is also trivially true. For the detailed proof please see Appendix \ref{app:classicalcorr}.

Consequently, the correlations obtained using projective simulable measurements on quantum systems with bounded local operational dimensions can be equivalently characterised using the correlations obtained from analogous classically correlated systems. We shall use this later to find the projective simulable bound on the figure of merit for some task $\mathbb{G}~(n,k,d)$. This observation leads us to the following.

\begin{cor}\label{cor:non-proj detect}
If there exists a quantum strategy using a multipartite pre-shared quantum state, with local dimension $d$, and a $k$-outcome measurement device such that $\mathcal{R}^{Q}[\mathbb{G}~(n,d,k)]>\mathcal{R}_{max}^{Cl}[\mathbb{G}~(n,d,k)]$, then the measurement device must be non-projective. 
\end{cor}

\begin{proof}
    For a task $\mathbb{G}~(n,d,k)$, we know from theorem \ref{theo:classicalcorr} we know that $\mathcal{C}^{Cl}_d (k) = \mathcal{C}^{PQ}_d (k)$. Consequently, $\mathcal{R}_{max}^{Cl}[\mathbb{G}~(n,d,k)]=\mathcal{R}_{max}^{PQ}[\mathbb{G}~(n,d,k)]$. Thus, the violation of the classical bound can be used to witness non-projective measurements if the local dimension of the pre-shared quantum system is bounded by $d$.
\end{proof}

Now we will look into certain specific tasks $\mathbb{G}~(n,k,d)$ to detect qubit non-projective measurements. In section \ref{sec:S223} and \ref{sec:S224} the task $\mathbb{G}~(2,2,3)$ and $\mathbb{G}~(2,2,4)$ are respectively used to detect qubit three and four-outcome non-projective measurements. In section \ref{sec:S223S} we use task $\mathbb{G}~(2,2,3)$ with an added constraint of having the two measuring devices to be identical. There we show that certain correlations are qubit non-projective simulable but cannot be obtained using two copies of shared two-qubit state and projective measurement on each copy followed by post-processing of the measurement outcomes.

It is worth noting the following regarding correlations obtained from projective simulable measurements when the measuring devices are assumed to be identical.  In the task $\mathbb{G}~(n,d,k)$ with this constraint on the measuring devices, the set of projective qudit simulable correlations is also the same as the classically simulable correlations generated from identical post-processing on a shared classical state.   Even if identical projective simulable measurements are performed on a shared $d$ dimensional quantum system, it is equivalent to assuming that only the post-processing of the outcomes from a projective measurement is identical for both parties. The shared state and measurement would change accordingly to keep the probability distribution of the outcome invariant. Using this observation and the proof of theorem \ref{theo:classicalcorr} provided in appendix \ref{app:classicalcorr} it is easy to see the following. The correlations generated using identical projective measurements on a shared quantum system of local dimension $d$ are the same as the correlations generated using $d$-level shared classical state and identical post-processing of measurement outcomes.

\subsection{Detecting $3$-outcome qubit non-projective Measurements: $\mathbb{G}~(2,2,3)$}\label{sec:S223}
For the task $\mathbb{G}~(2,2,3)$ we consider two parties Alice and Bob. ${\map P}_{AB}$ be a preparation device that prepares a bipartite state, say $\rho_{AB}$, with local operational dimension $2$ and distributes them between these two parties. Their respective measuring devices ${\map M}_A=\{E_i\}_{i=0}^2$ and ${\map M}_B=\{F_i\}_{i=0}^2$, which can yield $3$-outcomes each, outputs $a\in\{0,1,2\}$ and $b\in\{0,1,2\}$ respectively. Using this bipartite state and the measuring devices the parties observe joint probability of outcomes $\{p(a,b\}_{a,b}$. Here, $p(a,b)= Tr(\rho_{AB}(E_a\otimes F_b))$.  Now, we consider the following figure of merit:

\begin{equation}
{\map R}[\mathbb{G}~(2,2,3)]=\min_{a,b \atop b\ne a} p(a,b)
\end{equation}

The maximum value of the figure of merit is obtained when the probability of all correlated events is zero, {\it i.e.} $p(a,b)=0$ if $a=b$, and the anti-correlated events are equiprobable, {\it i.e.} $p(a,b)=\frac{1}{6}$ if $a\neq b$. For any other distribution, the least probable anti-correlated event will have a probability lower than $\frac{1}{6}$. Thus, the maximum payoff that can be obtained using any shared state and measurements for this task is $\frac{1}{6}$. Using this figure of merit, we will now state the following result that would be useful later:

\begin{theo}\label{theo:SRG1}
Given any pre-shared bipartite state with local operational dimension $2$, the maximum payoff using projective simulable measurements $${\map R}^{PQ}_{max}[\mathbb{G}~(2,2,3)]=\frac{1}{8}< {\map R}^{Q}_{max}[\mathbb{G}~(2,2,3)]=\frac{1}{6}.$$
\end{theo}

\begin{proof}
    From theorem \ref{theo:classicalcorr}, any correlation that can be obtained using projective simulable measurement on a bipartite qubit state can also be obtained from a two-level classically correlated system. Also from lemma \ref{lemma:generator2} {\it w.l.o.g.}, we can assume that the initial shared classically correlated state is $\omega_{A,B}^{Cl}=(\lambda_{0,0}=\lambda,\lambda_{0,1}=0,\lambda_{1,0}=0,\lambda_{1,1}=1-\lambda)$ where $\lambda\in[0,1]$. From theorem $1$ and $2$ in \cite{Guha2021}, we know that the payoff obtained using any two-level classically correlated state is sub-optimal and bounded by $\frac{1}{8}$. The violation of this bound can be used to detect qubit non-projective measurement. The quantum state and measurement that leads to the optimum payoff are shared singlet $\ket{\psi^-}_{AB}=\frac{1}{\sqrt{2}}(\ket{01}_{AB}-\ket{10}_{AB})$ and {\it trine} measurement ${\map M}_A={\map M}_B=\{\frac{2}{3}\Pi_{\alpha}=\ket{\psi_{\alpha}}\bra{\psi_{\alpha}}:\ket{\psi_{\alpha}}=\cos[\frac{2\pi}{3}\alpha]\ket{0}+  \sin[\frac{2\pi}{3}\alpha]\ket{1}, \alpha\in\{0,1,2\} \}$.  
\end{proof}

 Even when considering white noise acting on the shared bipartite state, {\it i.e.} $p\ket{\psi^-}\bra{\psi^-}+(1-p)\frac{\mathbb{I}}{2}\otimes\frac{\mathbb{I}}{2}$, the payoff obtained from the same measurement setting is higher than classical payoff when $p>\frac{1}{4}$. Thus, using the results from \cite{Guha2021} one can detect three-outcome qubit non-projective measurements whenever they violate the projective simulable bound. A similar result to detect four-outcome qubit non-projective measurement can be obtained for a bipartite task $\mathbb{G}~(2,2,4)$ while considering a non-convex figure of merit as in \cite{Guha2021}.

Instead of considering a figure of merit discussed above for the task $\mathbb{G}~(2,2,3)$, we can consider specific target correlations that need to be simulated as a winning condition. Here we will consider two different such sets of target correlations for the task $\mathbb{G}~(2,2,3)$. The first set of target correlation $T_1[2,2,3]=\{\{p(a,b)\}_{a,b=0}^2\}$ is of the following form 
 \begin{equation}\label{eq:corr t1 G223}
    p(a,b)=\begin{cases}
        & 0 ~~~~~~~~~~~~\text{ for } a=b\\
        & \alpha_{a,b}(>0) \text{ for } a\neq b
    \end{cases}
 \end{equation}
 The second set of target correlation $T_2[2,2,3]=\{\{p(a,b)\}_{a,b=0}^2\}$ for the task $\mathbb{G}~(2,2,3)$ is of the following form 
 \begin{equation}
    p(a,b)=\begin{cases}\label{eq:corr t2 G223}
        & \alpha \text{ for } a=b\\
        & \beta \text{ for } a\neq b
    \end{cases} \text{ where } \alpha\neq \beta
 \end{equation}
 In other words, we require that the correlated events have the same probability and the anti-correlated events have the same probability.
 
 For this correlation-simulation task, in the following, we show that it is impossible to obtain these kinds of correlations while using any qubit projective simulable measurements. Also, we will show that these target correlations can be used to detect three-outcome non-projective measurements. 
 \subsubsection{Simulability of target correlation $T_1[2,2,3]$}\label{sec: 223 t1}
 \begin{theo}\label{theo:classical223 noiseless}
    For the task $\mathbb{G}~(2,2,3)$, $$ \mathcal{C}^{PQ}_2 (3)\bigcap T_1[2,2,3]=\varnothing.$$ 
\end{theo}

\begin{proof}
    From theorem \ref{theo:classicalcorr}, we know that correlations obtained using qubit projective simulable measurements can also be obtained from a two-level classically correlated system. From lemma \ref{lemma:generator2}, {\it w.l.o.g.}, the initially shared classically correlated state be $\omega_{A,B}^{Cl}=(\lambda_{0,0}=\lambda,\lambda_{0,1}=0,\lambda_{1,0}=0,\lambda_{1,1}=1-\lambda)$ where $\lambda\in[0,1]$. Now using lemma $1$ and $2$ \footnote{For $N>2$, any correlation $\{p(a,b)\}_{a,b=0}^{N}$ such that $p(a,a)=0~\forall a\in\{0,1,\cdots N\}$ and $p(a,b)>0~\forall a,b(\neq a)\in\{0,1,\cdots N\}$ cannot be simulated by sharing only bipartite classical system of local dimension $2$.} in \cite{Guha2021}, it is straightforward to see that the correlations shown in eq.(\ref{eq:corr t1 G223}) cannot be generated using $\omega_{A,B}^{Cl}$ and performing a local stochastic operation by Alice and Bob.  
\end{proof}

\begin{theo}\label{theo:quant223 noiseless}
 For the task $\mathbb{G}~(2,2,3)$,
$$\mathcal{C}^{Q}_2 (3)\bigcap T_1[2,2,3] \neq \varnothing
$$
\end{theo}

\begin{proof}
    Let Alice and Bob pre-share the two qubit singlet state $\ket{\psi^-}_{AB}=\frac{1}{\sqrt{2}}(\ket{01}_{AB}-\ket{10}_{AB})$. Both perform the three-outcome extremal qubit POVM ${\map M}_A={\map M}_B=\{E_i=\lambda_i(\mathbb{I}+\hat{n}_i\cdot \vec{\sigma})\}_{i=0}^2$ where $\lambda_i>0$, $\sum_i \lambda_i=1$, $\vec{\sigma}=(\sigma_1,\sigma_2,\sigma_3)$ is a vector of Pauli matrices, $\hat{n}_i\in \mathbb{R}^3$, $|\hat{n}_i|^2=1$ (unit vector), $\sum_i \lambda_i \hat{n}_i=0$ and $\hat{n}_i\neq \hat{n}_j~\forall i,j$. The correlation obtained from the state and the measurement is $p(a,b)=\lambda_a\lambda_b(1-\hat{n}_a\cdot\hat{n}_b)$. This correlation is of the form described in eq.(\ref{eq:corr t1 G223}). 
\end{proof}

As a consequence of theorem \ref{theo:classical223 noiseless} and \ref{theo:quant223 noiseless}, the correlations in $T_1[2,2,3] $ can be used to detect three-outcome qubit non-projective measurements. However, this detection scheme can be used in the absence of any noise. Next, we will discuss a detection scheme that can also be used in the presence of arbitrary noise. 

 \subsubsection{Simulability of target correlation $T_2[2,2,3]$}\label{sec: 223 t2}

\begin{theo}\label{theo:classical223}
    For the task $\mathbb{G}~(2,2,3)$, $$ \mathcal{C}^{PQ}_2 (3)\cap T_2[2,2,3]=\varnothing.$$ 
\end{theo}

\begin{proof}
From theorem \ref{theo:classicalcorr}, all qubit projective simulable correlations can also be obtained from a two-level classically correlated system. Also from lemma \ref{lemma:generator2} {\it w.l.o.g.}, the initial shared classically correlated state be $\omega_{A,B}^{Cl}=(\lambda_{0,0}=\lambda,\lambda_{0,1}=0,\lambda_{1,0}=0,\lambda_{1,1}=1-\lambda)$ where $\lambda\in[0,1]$. A general local stochastic map for Alice is given by 

\begin{equation}
S^A_{2\rightarrow 3}:=(s_{lm})_{l=0,1,2\atop m=0,1}
\end{equation}

such that 
\begin{equation}\label{eq: classical223 1}
\sum_{l=0}^2 s_{lm}=1~\forall~m\in \{{0,1}\}\text{ and } s_{lm}\ge 0
\end{equation}
Similarly, a general local stochastic map for Bob will be 
\begin{equation}
S^B_{2\rightarrow 3}:=(s'_{lm})_{l=0,1,2\atop m=0,1}
\end{equation}

such that 
\begin{equation}\label{eq: classical223 2}
\sum_{l=0}^2 s'_{lm}=1 ~\forall~m\in \{{0,1}\} \text{ and } s'_{lm}\ge 0
\end{equation}

Now, the correlations obtained by them are given by 
\begin{equation}\label{eq: classical223 3}
P=(S_{2\rightarrow 3}^{A}\otimes S_{2\rightarrow 3}^{B})(\omega_{A,B}^{Cl})^T:= (p(a,b))_{a,b=0,1,2}
\end{equation}

To check if we can obtain any correlation in $T_2[2,2,3]$ using the two-level classically correlated system, we need to solve the following set of equations.

\begin{equation}\label{eq: classical223 4}
    p(0,0)=p(1,1)=p(2,2)\ne \frac{1}{9}
\end{equation}
\begin{equation}\label{eq: classical223 5}
p(0,1)=p(0,2)=p(1,0)=p(1,2)=p(2,0)=p(2,1)
\end{equation}

From eq.(\ref{eq: classical223 4}) and (\ref{eq: classical223 5}) we obtain that all the marginal correlation for Alice $p(a)=\frac{1}{3}~\forall a\{0,1,2\}$. Similarly, from eq.(\ref{eq: classical223 4}) and (\ref{eq: classical223 5}) we obtain that all the marginal correlation for Bob $p(b)=\frac{1}{3}\forall b\{0,1,2\}$. Using this with eq.(\ref{eq: classical223 1}, \ref{eq: classical223 2}, \ref{eq: classical223 3}) we obtain the following:

\begin{equation}
    s_{01}= \frac{1}{(1 - \lambda)} (\frac{1}{3} - \lambda s_{00}) \text{ and }  s_{11}= \frac{1}{(1 - \lambda)} (\frac{1}{3} - \lambda s_{10})
\end{equation}

\begin{equation}
    s'_{01}= \frac{1}{(1 - \lambda)} (\frac{1}{3} - \lambda s'_{00}) \text{ and }  s'_{11}= \frac{1}{(1 - \lambda)} (\frac{1}{3} - \lambda s'_{10})
\end{equation} 

Upon substituting the above and solving the linear set of eq. (\ref{eq: classical223 4}) and  (\ref{eq: classical223 5}) under the constraint given in eq.(\ref{eq: classical223 1}) and (\ref{eq: classical223 2}), assuming that either $p(0,0)>p(0,1)$ or $p(0,0)<p(0,1)$, we get that there is no solution to the system of equations.
\end{proof}

\begin{theo}\label{theo:quant223}
 For the task $\mathbb{G}~(2,2,3)$,
$$\varnothing\neq\mathcal{C}^{Q}_2 (3)\bigcap T_2[2,2,3]\subsetneq T_2[2,2,3]$$
\end{theo}

\begin{proof}
    When Alice and Bob share any qubit bipartite state they can only perform some local operation during the task. The mutual information is a non-increasing quantity under local trace-preserving maps and is upper-bounded by the logarithm of the local dimension of the shared system \cite{wilde2011classical}.  Therefore, the mutual information trivially imposes a bound on the correlations that can be obtained while sharing two two-qubit systems implying the following: any correlation in the set $T_2[2,2,3]$ whose mutual information is greater than $1$ cannot be simulated using a shared qubit bipartite system. If $p(0,0)=x$ and $p(0,1)=y$ then using the property of the correlations in the set $T_2[2,2,3]$ and normalization of probability, we get,
\begin{equation}
   3x+6y=1\implies x+2y=\frac{1}{3}
 \end{equation}
In this case, mutual information is given by $2\log_2 3+3x\log_2 x+6y\log_2 y$. Thus, even using qubit non-projective measurement we cannot obtain correlations in $T_2[2,2,3]$ such that 
\begin{equation}
2\log_2 3+3x\log_2 x+6y\log_2 y >1
\end{equation}    

This shows that there are some correlations in the set $T_2[2,2,3]$ that cannot be obtained using any qubit measurements (local) by Alice and Bob (region $R4$ in figure \ref{fig:Noisy223}). Now, we will show that Alice and Bob can obtain some correlations in the set $T_2[2,2,3]$ by performing local measurements on a bipartite two-qubit state. If Alice and Bob share the state $\ket{\psi^-}_{AB}=\frac{1}{\sqrt{2}}(\ket{01}_{AB}-\ket{10}_{AB})$ and they perform trine measurement ${\map M}_A={\map M}_B=\{\Pi_{\alpha}=\frac{2}{3}\ket{\psi_{\alpha}}\bra{\psi_{\alpha}}:\ket{\psi_{\alpha}}=\cos[\frac{2\pi}{3}\alpha]\ket{0}+  \sin[\frac{2\pi}{3}\alpha]\ket{1}, \alpha\in\{0,1,2\} \}$, Alice and Bob can obtain correlations such that $p(a,b)=0$ if $a=b$ and $p(a,b)=\frac{1}{6}$ otherwise.
\end{proof}

Three-outcome trine measurements on a two-qubit shared state can simulate more correlations in the set $T_2[2,2,3]$. Using the shared state $\rho_p:= p\ket{\psi^-}\bra{\psi^-}+(1-p)\frac{\mathbb{I}}{2}\otimes\frac{\mathbb{I}}{2}$, where $\ket{\psi^-}_{AB}=\frac{1}{\sqrt{2}}(\ket{01}_{AB}-\ket{10}_{AB})$ is the qubit singlet state and $p\in [0,1]$, and the trine measurement ${\map M}_A={\map M}_B=\{\Pi_{\alpha}=\frac{2}{3}\ket{\psi_{\alpha}}\bra{\psi_{\alpha}}:\ket{\psi_{\alpha}}=\cos[\frac{2\pi}{3}\alpha]\ket{0}+  \sin[\frac{2\pi}{3}\alpha]\ket{1}, \alpha\in\{0,1,2\} \}$, Alice and Bob can obtain correlations given in table \ref{table: noisy singlet-triene}. Thus, using the state and measurement specified, $x$ can take all values between $[0,\frac{1}{9}]$, limiting values being achieved for $p=1$ and $p=0$ respectively (region $R1$ in figure \ref{fig:Noisy223}).

    \begin{table}[h!]
        \centering
        \begin{tabular}{|c|c|c|c|}
        \hline
            $a\backslash b$ & $0$ & $1$ & $2$\\
            \hline
            $0$ & $\frac{1}{9}(1-p)$ & $\frac{1}{18}(2+p)$ & $\frac{1}{18}(2+p)$\\
            \hline
            $1$ & $\frac{1}{18}(2+p)$ & $\frac{1}{9}(1-p)$ & $\frac{1}{18}(2+p)$\\
            \hline
            $2$ & $\frac{1}{18}(2+p)$ & $\frac{1}{18}(2+p)$ & $\frac{1}{9}(1-p)$\\
            \hline
        \end{tabular}
        \caption{Joint probability $p(a,b)$ obtained for different amounts of noise while performing the trine measurement on shared qubit state $\rho_p$.}
        \label{table: noisy singlet-triene}
    \end{table}

 Similarly, using the shared state $\Tilde{\rho}_p:= p\ket{\psi^+}\bra{\psi^+}+(1-p)\frac{\mathbb{I}}{2}\otimes\frac{\mathbb{I}}{2}$, where $\ket{\psi^+}_{AB}=\frac{1}{\sqrt{2}}(\ket{01}_{AB}+\ket{10}_{AB})$ and $p\in [0,1]$, and the trine measurement ${\map M}_A={\map M}_B=\{\frac{2}{3}\Pi_{\alpha}=\ket{\psi_{\alpha}}\bra{\psi_{\alpha}}:\ket{\psi_{\alpha}}=\cos[\frac{2\pi}{3}\alpha]\ket{\frac{0+1}{\sqrt{2}}}+  \sin[\frac{2\pi}{3}\alpha]\ket{\frac{0-1}{\sqrt{2}}}, \alpha\in\{0,1,2\} \}$, Alice and Bob can obtain correlations given in table \ref{table: noisy psi+-triene}. Using the state and measurement specified,  $x$ can take all values between $[\frac{1}{9},\frac{2}{9}]$, limiting values being achieved for $p=0$ and $p=1$ respectively (region $R2$ in figure \ref{fig:Noisy223}).

    \begin{table}[h!]
        \centering
        \begin{tabular}{|c|c|c|c|}
        \hline
            $a\backslash b$ & $0$ & $1$ & $2$\\
            \hline
            $0$ & $\frac{1}{9}(1+p)$ & $\frac{1}{18}(2-p)$ & $\frac{1}{18}(2-p)$\\
            \hline
            $1$ & $\frac{1}{18}(2-p)$ & $\frac{1}{9}(1+p)$ & $\frac{1}{18}(2-p)$\\
            \hline
            $2$ & $\frac{1}{18}(2-p)$ & $\frac{1}{18}(2-p)$ & $\frac{1}{9}(1+p)$\\
            \hline
        \end{tabular}
        \caption{Joint probability $p(a,b)$ obtained for different amounts of noise while performing trine measurement on shared state $\Tilde{\rho}_p$.} 
        \label{table: noisy psi+-triene}
    \end{table}

\begin{figure}[h!]
    \centering
    \includegraphics[scale=0.4]{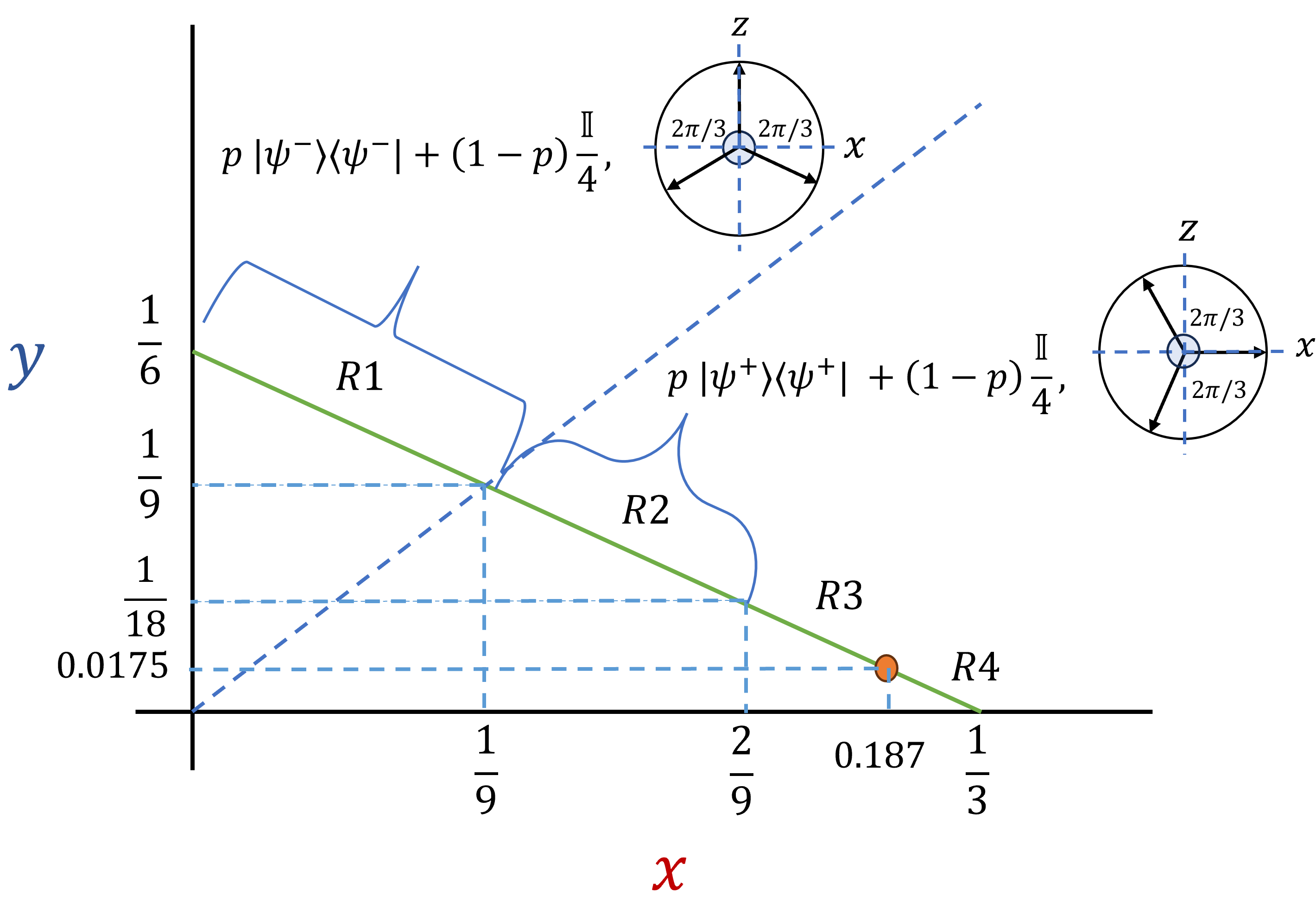}
    \caption{The correlations in the set $T_2[2,2,3]$ that can be obtained using qubit non-projective measurement. Here $x,y$ denotes the probability of each event where the outcomes are correlated and anti-correlated respectively. Correlations in the region $R1$ and $R2$ can be obtained using trine local measurements on the shared qubit state $\rho_p$ (mixture of maximally entangled state $|\psi^{-}\rangle$ and maximally mixed state) and rotated trine local measurements on $\Tilde{\rho}_p$ (mixture of maximally entangled state $|\psi^{+}\rangle$ and maximally mixed state) respectively. The correlations in the region $R4$ cannot be obtained using any local quantum measurements on a two-qubit shared state. The realizability of the correlations in the region $R3$ with shared states of local dimension $2$ is unknown.}
    \label{fig:Noisy223}
\end{figure}

Note that, firstly, the aforementioned correlations $\mathcal{C}^{Q}_2 (3)\bigcap T_2[2,2,3]$ are obtained using some non-projective simulable measurement on shared bipartite qubit state. This follows as a consequence of theorem \ref{theo:classical223} and \ref{theo:quant223}. Secondly, the correlations obtained from the qubit non-projective measurements shown in the table \ref{table: noisy singlet-triene} and \ref{table: noisy psi+-triene} do not exhaust all the correlations in the set $T_2[2,2,3]$ for which mutual information is less than $1$. In other words, whether the correlations such that $x>\frac{2}{9}$, while the mutual information is less than or equal to $1$, (see correlations marked by $R3$ in figure \ref{fig:Noisy223}) is achievable using qubit non-projective measurement needs to be explored further.

\subsubsection{Robustness against noise}
As discussed in theorem \ref{theo:classical223} the set of correlations in $T_2[2,2,3]$ is not qubit projective simulable while some of them can be obtained using shared two-qubit state and local non-projective measurement. If Alice and Bob want to simulate the correlation $\{p(a,b)\}_{a,b=0}^{2}\in T_2[2,2,3]$ such that $p(i,i)=0~\forall~i\in\{0,1,2\}$ then they can use the shared two qubit state $\ket{\psi^-}_{AB}=\frac{1}{\sqrt{2}}(\ket{01}_{AB}-\ket{10}_{AB})$ and perform the trine measurement ${\map M}_A={\map M}_B=\{\frac{2}{3}\Pi_{\alpha}=\ket{\psi_{\alpha}}\bra{\psi_{\alpha}}:\ket{\psi_{\alpha}}=\cos[\frac{2\pi}{3}\alpha]\ket{0}+  \sin[\frac{2\pi}{3}\alpha]\ket{1}, \alpha\in\{0,1,2\} \}$. However, this shared state may become noisy. Considering depolarising noise here the evolved shared state becomes $\rho_{\epsilon_s}:= \epsilon_s\ket{\psi^-}\bra{\psi^-}+(1-\epsilon_s)\frac{\mathbb{I}}{2}\otimes\frac{\mathbb{I}}{2}$ where $\epsilon_s\in (0,1]$. Also, the measurement device can be noisy. Here  we specifically consider depolarising noise acting on the measurement device and thus the effective measurement ${\map M}_X$, where $X=A,B$ for Alice and Bob respectively, is given as follows:

\begin{equation}
\begin{split}
& \Pi_{0}=\begin{pmatrix}
\frac{1+\epsilon_X}{3} & 0 \\
0 & \frac{1-\epsilon_X}{3} 
\end{pmatrix},
\Pi_{1}=\begin{pmatrix}
\frac{2-\epsilon_X}{6} & \frac{-\epsilon_X}{2\sqrt{3}} \\
\frac{-\epsilon_X}{2\sqrt{3}} & \frac{2+\epsilon_X}{6} 
\end{pmatrix},
\Pi_{2}=\begin{pmatrix}
\frac{2-\epsilon_X}{6} & \frac{\epsilon_X}{2\sqrt{3}} \\
\frac{\epsilon_X}{2\sqrt{3}} & \frac{2+\epsilon_X}{6} 
\end{pmatrix} \\
\end{split}
\end{equation}

 Here $\epsilon_A,\epsilon_B\in(0,1]$. The correlation obtained using this noisy state and measurement is given in table \ref{table: noisy singlet-noisy triene}. Interestingly for arbitrary amounts of noise in state and measurement, the correlation obtained belongs to $T_2[2,2,3]$. Thus, the detection scheme proposed here is robust against arbitrary depolarising noise.

 \begin{table}[h!]
        \centering
        \begin{tabular}{|c|c|c|c|}
        \hline
            $a\backslash b$ & $0$ & $1$ & $2$\\
            \hline
            $0$ & $\frac{1}{9}(1-\epsilon_s \epsilon_A \epsilon_B)$ & $\frac{1}{18}(2+\epsilon_s \epsilon_A \epsilon_B)$ & $\frac{1}{18}(2+\epsilon_s \epsilon_A \epsilon_B)$\\
            \hline
            $1$ & $\frac{1}{18}(2+\epsilon_s \epsilon_A \epsilon_B)$ & $\frac{1}{9}(1-\epsilon_s \epsilon_A \epsilon_B)$ & $\frac{1}{18}(2+\epsilon_s \epsilon_A \epsilon_B)$\\
            \hline
            $2$ & $\frac{1}{18}(2+\epsilon_s \epsilon_A \epsilon_B)$ & $\frac{1}{18}(2+\epsilon_s \epsilon_A \epsilon_B)$ & $\frac{1}{9}(1-\epsilon_s \epsilon_A \epsilon_B)$\\
            \hline
        \end{tabular}
        \caption{Joint probability $p(a,b)$ obtained from noisy state $\rho_{\epsilon_s}$ and noisy trine measurement on the shared qubit state. Here, $\epsilon_s, \epsilon_A,\epsilon_B$ are the parameters corresponding to the noise in the pre-shared state, Alice's and Bob's measuring device respectively.}
        \label{table: noisy singlet-noisy triene}
    \end{table}

Alternately, Alice and Bob may want to simulate the correlation $\{p(a,b)\}_{a,b=0}^{2}\in T_2[2,2,3]$ such that $p(i,i)=\frac{2}{9}~\forall~i\in\{0,1,2\}$. Then they can use the shared two qubit state $\ket{\psi^+}_{AB}=\frac{1}{\sqrt{2}}(\ket{01}_{AB}+\ket{10}_{AB})$ and perform the trine measurement ${\map M}_A={\map M}_B=\{\frac{2}{3}\Pi_{\alpha}=\ket{\psi_{\alpha}}\bra{\psi_{\alpha}}:\ket{\psi_{\alpha}}=\cos[\frac{2\pi}{3}\alpha]\ket{\frac{0+1}{\sqrt{2}}}+  \sin[\frac{2\pi}{3}\alpha]\ket{\frac{0-1}{\sqrt{2}}}, \alpha\in\{0,1,2\} \}$. Analogous to the previous case, it is possible to show even for this state and measurement, the robustness of the detection scheme against arbitrary depolarising noise.

In general, if Alice and Bob share an arbitrary two-qubit state $\rho_{AB}$, they can perform three-outcome qubit non-projective measurements, denoted as $\mathcal{M}_A = \{E_i\}_{i=0}^{2}$ and $\mathcal{M}_B = \{F_i\}_{i=0}^{2}$, where $\sum_i E_i = \sum_i F_i = \mathbb{I}$, $E_i > 0$ and $F_i > 0$. The correlations they observe are given by $Tr(\rho_{AB}(E_a\otimes F_b))$. The set of non-projective measurements that can be detected using a fixed shared state is obtained by solving the equation $Tr(\rho_{AB}(E_a\otimes F_b))= p(a,b)$ where the correlation $p(a,b)$ is defined in eq.(\ref{eq:corr t2 G223}). For a given state this is a set of quadratic equations in the parameters of the POVMs.

\subsection{Detecting $4$-outcome qubit non-projective Measurements: $\mathbb{G}~(2,2,4)$}\label{sec:S224}
Now we will discuss a task to witness four-outcome non-projective measurements. Here also we consider two non-communicating parties Alice and Bob. Let ${\map P}_{AB}$ be a preparation device that prepares a bipartite state with local operational dimension $2$ and distributes them between these two parties. Their respective measuring devices ${\map M}_A$ and ${\map M}_B$, which can yield $4$-outcomes each, outputs $a\in\{0,1,2,3\}$ and $b\in\{0,1,2,3\}$ respectively. 
We will consider two different such sets of target correlations for the task $\mathbb{G}~(2,2,4)$. The first set of target correlation $T_1[2,2,4]=\{\{p(a,b)\}_{a,b=0}^3\}$ is of the following form 
 \begin{equation}\label{eq:corr t1 G224}
    p(a,b)=\begin{cases}
        & 0 ~~~~~~~~~~~~\text{ for } a=b\\
        & \alpha_{a,b}(>0) \text{ for } a\neq b
    \end{cases}
 \end{equation}
 The second set of target correlation $T_2[2,2,4]=\{\{p(a,b)\}_{a,b=0}^3\}$ for the task $\mathbb{G}~(2,2,4)$ is of the following form 
 \begin{equation}\label{eq:corr t2 G224}
    p(a,b)=\begin{cases}
        & \alpha \text{ for } a=b\\
        & \beta \text{ for } a\neq b
    \end{cases} \text{ where } \alpha\neq \beta
 \end{equation}

 Next, we show that any correlation in $T_1[2,2,4]$ and $T_2[2,2,4]$ cannot be obtained using some qubit projective simulable measurement on a pre-shared bipartite two-qubit state.

\subsubsection{Simulability of target correlation $T_1[2,2,4]$}\label{sec: 224 t1}
\begin{theo}\label{theo:classical224 noiseless}
For the task $\mathbb{G}~(2,2,4)$, $$ \mathcal{C}^{PQ}_2 (4)\cap T_1[2,2,4]=\varnothing$$ 
\end{theo}

\begin{proof}
    The proof is exactly similar to the proof of theorem \ref{theo:classical223 noiseless}.
\end{proof}

\begin{theo}\label{theo:quant224 noiseless}
 For the task $\mathbb{G}~(2,2,4)$,
$$\mathcal{C}^{Q}_2 (4)\bigcap T_1[2,2,4] \neq \varnothing$$
\end{theo}

\begin{proof}
    Let Alice and Bob pre-share the two qubit singlet state $\ket{\psi^-}_{AB}=\frac{1}{\sqrt{2}}(\ket{01}_{AB}-\ket{10}_{AB})$. And both perform the four-outcome extremal qubit POVM ${\map M}_A={\map M}_B=\{E_i=\lambda_i(\mathbb{I}+\hat{n}_i\cdot \vec{\sigma})\}_{i=0}^3$ where $\lambda_i>0$, $\sum_i \lambda_i=1$, $\vec{\sigma}=(\sigma_1,\sigma_2,\sigma_3)$ is a vector of Pauli matrices, $\hat{n}_i\in \mathbb{R}^3$, $|\hat{n}_i|^2=1$ (unit vector), $\sum_i \lambda_i \hat{n}_i=0$ and $\hat{n}_i\neq \hat{n}_j~\forall i,j$. The correlation obtained from the state and the measurement is $p(a,b)=\lambda_a\lambda_b(1-\hat{n}_a\cdot\hat{n}_b)$. This correlation is of the form described in eq.(\ref{eq:corr t1 G224}). 
\end{proof}

Therefore, the correlations in $T_1[2,2,4]$ can be used to detect three-outcome qubit non-projective measurements. However, this detection scheme can be used in the absence of any noise. Next, we will discuss a detection scheme that can also be used in the presence of arbitrary depolarising noise. 
 
\subsubsection{Simulability of target correlation $T_2[2,2,4]$}\label{sec: 224 t2}
\begin{theo}\label{theo:classical224}
For the task $\mathbb{G}~(2,2,4)$, 
$$ \mathcal{C}^{PQ}_2 (4)\cap T_2[2,2,4]=\varnothing$$ 
\end{theo}
\begin{proof}
From theorem \ref{theo:classicalcorr} we can analyse the correlations that can be obtained using a two-level bipartite classically correlated state. And using lemma \ref{lemma:generator2} the initially shared classically correlated state {\it w.l.o.g.} can be $\omega_{A,B}^{Cl}=(\lambda,0,0,1-\lambda)$ where $\lambda\in[0,1]$. A general local stochastic map for Alice is denoted as

\begin{equation}
S^A_{2\rightarrow 4}:=(s_{lm})_{l=0,1,2,3\atop m=0,1}
\end{equation}

such that 
\begin{equation}\label{eq: classical224 1}
\sum_{l=0}^2 s_{lm}=1~\forall~m\in \{{0,1}\}\text{ and } s_{lm}\ge 0
\end{equation}
Similarly, a generic local stochastic map for Bob is denoted as
\begin{equation}
S^B_{2\rightarrow 4}:=(s'_{lm})_{l=0,1,2,3\atop m=0,1}
\end{equation}

such that 
\begin{equation}\label{eq: classical224 2}
\sum_{l=0}^2 s'_{lm}=1 ~\forall~m\in \{{0,1}\} \text{ and } s'_{lm}\ge 0
\end{equation}

The correlation obtained from the state after applying the local stochastic maps is given by 
\begin{equation}\label{eq: classical224 3}
P=(S_{2\rightarrow 3}^{A}\otimes S_{2\rightarrow 3}^{B})(\omega_{A,B}^{Cl})^T:= (p(a,b))_{a,b=0,1,2}
\end{equation}

To check if we can obtain any correlation in $T_2[2,2,4]$ using a two-level classically correlated system, we need to solve the following set of equations using the expressions obtained from eq.(\ref{eq: classical224 3}).

\begin{equation}\label{eq: classical224 4}
    p(0,0)=p(1,1)=p(2,2)=p(3,3)
\end{equation}

\begin{equation}\label{eq: classical224 5}
\begin{split}
& p(0,1)=p(0,2)=p(0,3)=p(1,0)=p(1,2)= p(1,3)=p(2,0)=p(2,1)=p(2,3)
\end{split}
\end{equation}

From eq.(\ref{eq: classical224 4}) and (\ref{eq: classical224 5}), we obtain that all the marginal correlation for Alice $p(a)=\frac{1}{4}~\forall a$. Similarly, from eq.(\ref{eq: classical224 4}) and (\ref{eq: classical224 5}), we obtain that all the marginal correlation for Bob $p(b)=\frac{1}{4}\forall b$. Using this with eq.(\ref{eq: classical224 1}), (\ref{eq: classical224 2}) and (\ref{eq: classical224 3}), we obtain the following:

\begin{equation}\label{eq: classical224 6} 
    s_{l1}= \frac{1}{1 - \lambda} (\frac{1}{4} - \lambda s_{l0})~\forall~l\in\{0,1,2,3\}
\end{equation}

\begin{equation}\label{eq: classical224 7}
   s'_{l1}= \frac{1}{1 - \lambda} (\frac{1}{4} - \lambda s'_{l0})~\forall~l\in\{0,1,2,3\}
\end{equation} 
Now upon substituting the eq.(\ref{eq: classical224 6}), (\ref{eq: classical224 7}) and solving the linear set of eq.(\ref{eq: classical224 4}) and (\ref{eq: classical224 5}) under the constraint given in eq.(\ref{eq: classical224 1}) and (\ref{eq: classical224 2}), we get that there is only one solution for which $p(a,b)=\frac{1}{16}~\forall~a,b$. There is no solution such that either $p(0,0)>p(0,1)$ or $p(0,0)<p(0,1)$.

\end{proof}

\begin{theo}\label{theo:quant224}
For the task $\mathbb{G}~(2,2,4)$,
$$\varnothing\neq\mathcal{C}^{Q}_2 (4)\bigcap T_2[2,2,4] \subsetneq T_2[2,2,4]$$
\end{theo}

\begin{proof}
    First, we shall show that some of the correlations in the set $T_2[2,2,4]$ cannot be obtained using any local quantum measurement performed on a bipartite two-qubit state. When Alice and Bob share any qubit bipartite state they can only perform some local operation during the task. The mutual information being a non-increasing quantity under local trace-preserving maps is upper-bounded by the logarithm of the local dimension of the shared system \cite{wilde2011classical}. Thus, any correlation in the set $T_2[2,2,4]$ whose mutual information is greater than $1$ cannot be simulated while sharing a bipartite two-qubit system. Say, $p(0,0)=x$ and $p(0,1)=y$ then using the property of the correlations in the set $T_2[2,2,4]$ and normalization of probability, we get,

    \begin{equation}
        4x+12y=1\implies x+3y=\frac{1}{4}
    \end{equation}
    
In this case, mutual information is given by $2\log_2 4+4x\log_2x + 12y\log_2y$. Thus, even using qubit non-projective measurement we cannot obtain correlations in $T_2[2,2,4]$ such that $2\log_2 4+4x\log_2x + 12y\log_2y >1$ (region $R4$ in figure \ref{fig:Noisy224}). 
    
Now, we shall show explicitly that some of the correlations in set $T_2[2,2,4]$ can be obtained using some local quantum measurements by Alice and Bob on their pre-shared state two-qubit state. If the shared two-qubit state is $\ket{\psi^-}_{AB}=\frac{1}{\sqrt{2}}(\ket{01}_{AB}-\ket{10}_{AB})$ and the measurement is qubit SIC-POVM ${\map M}_A={\map M}_B=\{\Pi_{\alpha}=\frac{1}{2}\ket{\psi_{\alpha}}\bra{\psi_{\alpha}}:\ket{\psi_{0}}=\ket{0},\ket{\psi_{\alpha}}=\sqrt{\frac{1}{3}}\ket{0}+  \sqrt{\frac{2}{3}}e^{\frac{2\pi i}{3}(\alpha-1)}\ket{1} \text{ for } \alpha\in\{1,2,3\} \}$, Alice and Bob can obtain correlations such that $p(a,b)=0$ if $a=b$ and $p(a,b)=\frac{1}{12}$ otherwise.

\end{proof}

The qubit SIC-POVM can generate more correlations in the set $T_2[2,2,4]$ beyond those discussed in the proof above. Using the state $\rho_p:= p\ket{\psi^-}\bra{\psi^-}+(1-p)\frac{\mathbb{I}}{2}\otimes\frac{\mathbb{I}}{2}$, where $\ket{\psi^-}_{AB}=\frac{1}{\sqrt{2}}(\ket{01}_{AB}-\ket{10}_{AB})$ and $p\in [0,1]$, and perfoming qubit SIC-POVM ${\map M}_A={\map M}_B=\{\Pi_{\alpha}=\frac{1}{2}\ket{\psi_{\alpha}}\bra{\psi_{\alpha}}:\ket{\psi_{0}}=\ket{0},\ket{\psi_{\alpha}}=\sqrt{\frac{1}{3}}\ket{0}+  \sqrt{\frac{2}{3}}e^{\frac{2\pi i}{3}(\alpha-1)}\ket{1} \text{ for } \alpha\in\{1,2,3\} \}$, Alice and Bob can obtain correlations given in table \ref{table: noisy singlet-SIC}. Using this state and measurement specified above $x$ can take all values between $[0,\frac{1}{16}]$, extreme values being achieved for $p=1$ and $p=0$ respectively (region $R1$ in figure \ref{fig:Noisy224}).

Similarly, sharing the two-qubit state $\Tilde{\rho}_p:= p\ket{\phi^+}\bra{\phi^+}+(1-p)\frac{\mathbb{I}}{2}\otimes\frac{\mathbb{I}}{2}$, where $\ket{\phi^+}_{AB}=\frac{1}{\sqrt{2}}(\ket{00}_{AB}+\ket{11}_{AB})$ and $p\in [0,1]$, and by performing the qubit SIC-POVM mentioned below, Alice and  Bob can obtain correlations given in table \ref{table: noisy singlet-SIC 2}. Using this state and measurement specified above $x$ can take all values between $[\frac{1}{16},\frac{2}{16}]$, extreme values being achieved for $p=0$ and $p=1$ respectively (region $R2$ in figure \ref{fig:Noisy224}). 

\begin{equation}
  \begin{split}
      {\map M}_A=&\{\Pi_{\alpha}=\frac{1}{2}\ket{\psi_{\alpha}}\bra{\psi_{\alpha}}:\ket{\psi_{0}}=\frac{1}{\sqrt{2}}(\ket{0}+\ket{1}),\ket{\psi_{\alpha}}=\frac{(1+\sqrt{2}e^{\frac{2\pi i}{3}(\alpha-1)})\ket{0}+ (1-\sqrt{2}e^{\frac{2\pi i}{3}(\alpha-1)})\ket{1}}{\sqrt{6}}\\
      &\text{ where } \alpha\in\{1,2,3\} \}
  \end{split}  
\end{equation}
\begin{equation}
    \begin{split}
        {\map M}_B=&\{\Pi_{\alpha}=\frac{1}{2}\ket{\psi_{\alpha}}\bra{\psi_{\alpha}}:\ket{\psi_{0}}=\frac{1}{\sqrt{2}}(\ket{0}+\ket{1}), \ket{\psi_{\alpha}}=\frac{(1+\sqrt{2}e^{\frac{2\pi i}{3}(1-\alpha)})\ket{0}+ (1-\sqrt{2}e^{\frac{2\pi i}{3}(1-\alpha)})\ket{1}}{\sqrt{6}}\\
    &\text{ where } \alpha\in\{1,2,3\} \}
    \end{split}
\end{equation}

    \begin{table}[h!]
        \centering
        \begin{tabular}{|c|c|c|c|c|}
        \hline
            $a\backslash b$ & $0$ & $1$ & $2$ & $3$\\
        \hline    
            $0$ & $\frac{1}{16}(1-p)$ & $\frac{1}{48}(3+p)$ & $\frac{1}{48}(3+p)$ & $\frac{1}{48}(3+p)$\\
        \hline    
            $1$ & $\frac{1}{48}(3+p)$ & $\frac{1}{16}(1-p)$ & $\frac{1}{48}(3+p)$ & $\frac{1}{48}(3+p)$\\
        \hline    
            $2$ & $\frac{1}{48}(3+p)$ & $\frac{1}{48}(3+p)$ & $\frac{1}{16}(1-p)$ & $\frac{1}{48}(3+p)$\\
        \hline    
            $3$ & $\frac{1}{48}(3+p)$ & $\frac{1}{48}(3+p)$ & $\frac{1}{48}(3+p)$ & $\frac{1}{16}(1-p)$\\
        \hline    
        \end{tabular}
        \caption{Joint probability $p(a,b)$ obtained for different amounts of noise while performing SIC-POVM measurement on shared state $\rho_p$.}
        \label{table: noisy singlet-SIC}
    \end{table}

 \begin{table}[h!]
        \centering
        \begin{tabular}{|c|c|c|c|c|}
        \hline
            $a\backslash b$ & $0$ & $1$ & $2$ & $3$\\
        \hline    
            $0$ & $\frac{1}{16}(1+p)$ & $\frac{1}{48}(3-p)$ & $\frac{1}{48}(3-p)$ & $\frac{1}{48}(3-p)$\\
        \hline    
            $1$ & $\frac{1}{48}(3-p)$ & $\frac{1}{16}(1+p)$ & $\frac{1}{48}(3-p)$ & $\frac{1}{48}(3-p)$\\
        \hline    
            $2$ & $\frac{1}{48}(3-p)$ & $\frac{1}{48}(3-p)$ & $\frac{1}{16}(1+p)$ & $\frac{1}{48}(3-p)$\\
        \hline    
            $3$ & $\frac{1}{48}(3-p)$ & $\frac{1}{48}(3-p)$ & $\frac{1}{48}(3-p)$ & $\frac{1}{16}(1+p)$\\
        \hline    
        \end{tabular}
        \caption{Joint probability $p(a,b)$ obtained for different amounts of noise while performing SIC-POVM measurement on shared state $\Tilde{\rho}_p$.}
        \label{table: noisy singlet-SIC 2}
    \end{table}

\begin{figure}[h!]
    \centering
    \includegraphics[scale=0.1]{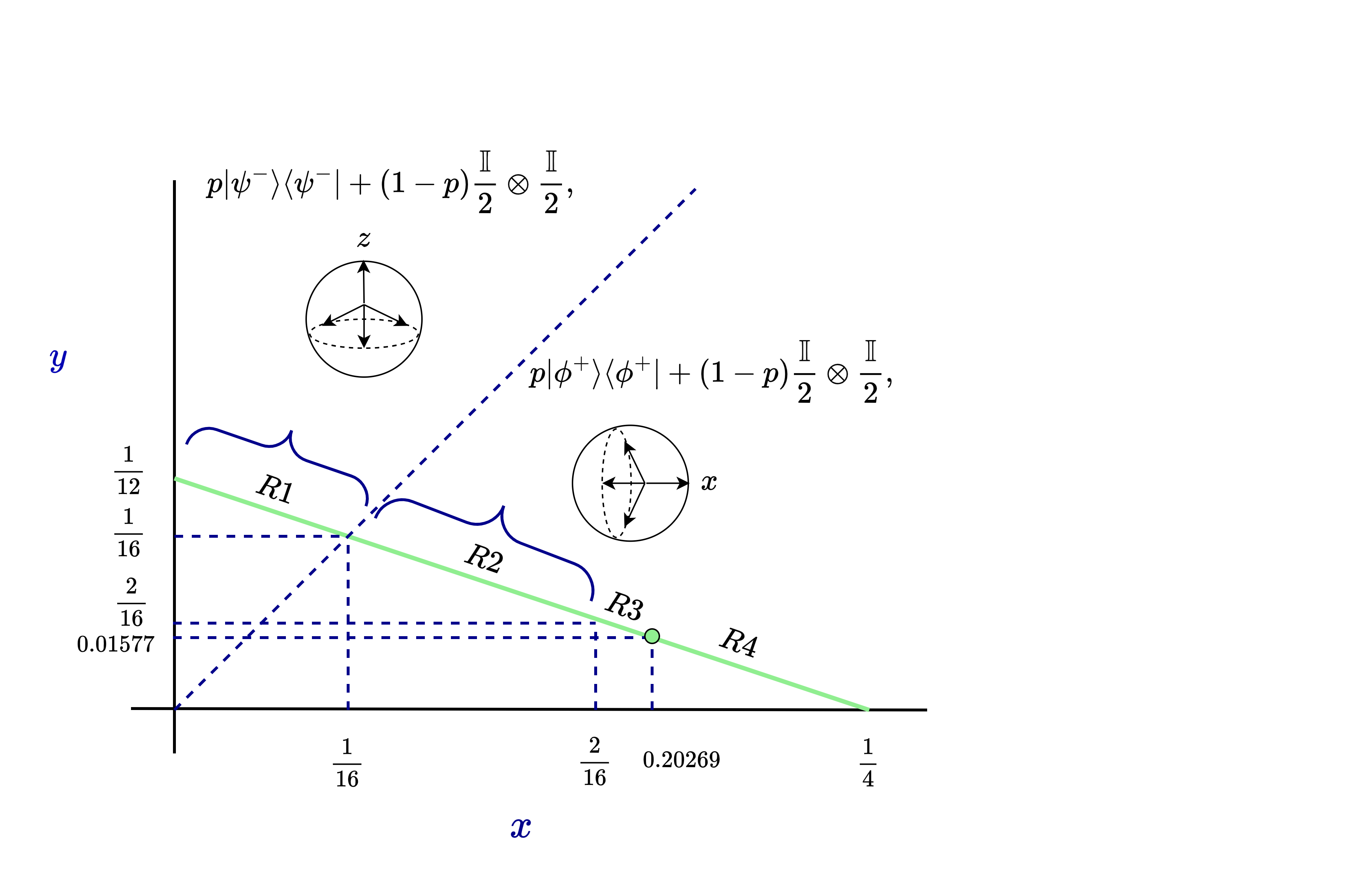}
\caption{The graph shows the correlations in the set $T_2[2,2,4]$ that can be obtained using qubit non-projective measurement. Here $x,y$ denotes the probability of each event where the outcomes are correlated and anti-correlated respectively. Correlations in the region $R1$ and $R2$ can be obtained using SIC-POVM local measurements on the shared two-qubit state $\rho_p$ (mixture of maximally entangled state $|\psi^{-}\rangle$ and maximally mixed state) and rotated SIC-POVM local measurements on $\Tilde{\rho}_p$ (mixture of maximally entangled state $|\phi^{+}\rangle$ and maximally mixed state) respectively. The correlations in the region $R4$ cannot be obtained using any local quantum measurements on a two-qubit shared state. The realizability of the correlations in the region $R3$ with shared states of local dimension $2$ is unknown.} 

    \label{fig:Noisy224}
\end{figure}

Firstly, all the correlations in the set $\mathcal{C}^{Q}_2 (4)\bigcap T_2[2,2,4]$ can obtained using only qubit non-projective simulable measurements on a shared quantum state. This follows directly as a consequence of theorem \ref{theo:classical224} and \ref{theo:quant224}. Secondly, the qubit state and measurement described above only achieve correlation such that $x\in [0,\frac{2}{16}]$. Whether all correlations in $T_2[2,2,4]$ whose mutual information is less than or equal to $1$ (see correlations marked by $R3$ in figure \ref{fig:Noisy224}) can be simulated using qubit non-projective measurement is to be explored further. 

\subsubsection{Robustness against noise}
From theorem \ref{theo:classical224}, no correlation in the set $T_2[2,2,4]$ can be obtained using qubit projective simulable measurements by Alice and Bob. However, they can obtain some of the correlations in this set using qubit non-projective measurements. If they want to simulate the correlation $\{p(a,b)\}_{a,b}\in T_2[2,2,4]$ such that $p(i,i)=0~\forall ~i\in \{0,1,2,3\}$ then they can use the pre-shared state $\ket{\psi^-}_{AB}=\frac{1}{\sqrt{2}}(\ket{01}_{AB}-\ket{10}_{AB})$ and perform the qubit SIC-POVM ${\map M}_A={\map M}_B=\{\Pi_{\alpha}=\frac{1}{2}\ket{\psi_{\alpha}}\bra{\psi_{\alpha}}:\ket{\psi_{0}}=\ket{0},\ket{\psi_{\alpha}}=\sqrt{\frac{1}{3}}\ket{0}+  \sqrt{\frac{2}{3}}e^{\frac{2\pi i}{3}(\alpha-1)}\ket{1} \text{ for } \alpha\in\{1,2,3\} \}$. However, the pre-shared state can become noisy. Considering depolarising noise acting on the state, the final state can be written as $\rho_{\epsilon_s}:= \epsilon_s\ket{\psi^-}\bra{\psi^-}+(1-\epsilon_s)\frac{\mathbb{I}}{2}\otimes\frac{\mathbb{I}}{2}$ where $\epsilon_s\in (0,1]$. Also, the measurement device can be noisy. Here,  we specifically consider that depolarising noise acts on the measurement device. The effective measurement ${\map M}_X$, where $X=A,B$ for Alice and Bob respectively, is given as follows:

\begin{equation}
\begin{split}
& \Pi_{0}=\begin{pmatrix}
\frac{1+\epsilon_X}{4} & 0 \\
0 & \frac{1-\epsilon_X}{4} 
\end{pmatrix},
\Pi_{1}=\begin{pmatrix}
\frac{3-\epsilon_X}{12} & \frac{\epsilon_X}{3\sqrt{2}} \\
\frac{\epsilon_X}{3\sqrt{2}} & \frac{3+\epsilon_X}{12} 
\end{pmatrix}, \Pi_{2}=\begin{pmatrix}
\frac{3-\epsilon_X}{12} & \frac{i^{\frac{1}{3}}\epsilon_X}{3\sqrt{2}} \\
\frac{i(i+\sqrt{3})\epsilon_X}{6\sqrt{2}} & \frac{3+\epsilon_X}{12}\\
\end{pmatrix},
\Pi_{3}=\begin{pmatrix}
\frac{3-\epsilon_X}{12} & \frac{i(i+\sqrt{3})\epsilon_X}{6\sqrt{2}} \\
-\frac{i^{\frac{1}{3}}\epsilon_X}{3\sqrt{2}} & \frac{3+\epsilon_X}{12} 
\end{pmatrix} \\
\end{split}
\end{equation}

Here $\epsilon_A,\epsilon_B\in(0,1]$. Using this noisy state and measurement, Alice and Bob obtain the correlation given in the table \ref{table: noisy singlet-noisy sic povm}. Note that this correlation belongs to $T_2[2,2,4]$ for arbitrary non-zero values of parameters $\epsilon_S,\epsilon_A,\epsilon_B$. Thus, the detection scheme proposed here is robust against arbitrary depolarising noise. Following analogous treatment as here, robustness against arbitrary noise can be shown if Alice and Bob want to simulate correlation $\{p(a,b)\}_{a,b}\in T_2[2,2,4]$ such that $p(i,i)=\frac{1}{8}~\forall ~i\in \{0,1,2,3\}$.  

    \begin{table}[h!]
        \centering
        \begin{tabular}{|c|c|c|c|c|}
        \hline
            $a\backslash b$ & $0$ & $1$ & $2$ & $3$\\
        \hline    
            $0$ & $\frac{1}{16}(1-\epsilon)$ & $\frac{1}{48}(3+\epsilon)$ & $\frac{1}{48}(3+\epsilon)$ & $\frac{1}{48}(3+\epsilon)$\\
        \hline    
            $1$ & $\frac{1}{48}(3+\epsilon)$ & $\frac{1}{16}(1-\epsilon)$ & $\frac{1}{48}(3+\epsilon)$ & $\frac{1}{48}(3+\epsilon)$\\
        \hline    
            $2$ & $\frac{1}{48}(3+\epsilon)$ & $\frac{1}{48}(3+\epsilon)$ & $\frac{1}{16}(1-\epsilon)$ & $\frac{1}{48}(3+\epsilon)$\\
        \hline    
            $3$ & $\frac{1}{48}(3+\epsilon)$ & $\frac{1}{48}(3+\epsilon)$ & $\frac{1}{48}(3+\epsilon)$ & $\frac{1}{16}(1-\epsilon)$\\
        \hline    
        \end{tabular}
        \caption{Joint probability $p(a,b)$ obtained from noisy state $\rho_{\epsilon_s}$ and noisy qubit SIC-POVM measurement. Here, $\epsilon=\epsilon_s \epsilon_A\epsilon_B$ and $\epsilon_s, \epsilon_A,\epsilon_B$ are the parameters corresponding to the noise in the pre-shared state, Alice's and Bob's measuring device respectively.}
        \label{table: noisy singlet-noisy sic povm}
    \end{table}

\subsection{$\mathbb{G}^{sym}~(2,2,3)$}\label{sec:S223S}
For the task  $\mathbb{G}~(2,2,3)$ discussed in the sub-section \ref{sec:S223}, we showed that no correlation in the set $T_2[2,2,3]$ can be obtained by qubit projective simulable measurement or analogously by any two level classically correlated system shared between Alice and Bob. About this task, a natural question is regarding the operational dimension required to simulate some specific correlation in $T_2[2,2,3]$ using projective simulable measurement only. For the task $\mathbb{G}^{sym}~(2,2,3)$, we assume Alice and Bob both have access to {\it identical copies} of an uncharacterised projective simulable measuring device as well as two bipartite preparations each of local dimension $2$ which can only be used adaptively. This is to say the parties cannot perform entangled measurements on the two local subsystems. Alternatively, they can perform a projective measurement on each local subsystem and jointly post-process the outcomes obtained from them.  Here we define the target correlation set to be $T[2,2,3]$ which contains a single correlation specified as $p(0,0)=p(1,1)=p(2,2)=0$ and $p(0,1)=p(0,2)=p(1,0)=p(1,2)=p(2,0)=p(2,1)=\frac{1}{6}$. Here we show that Alice and Bob while using {\it identical copies} of an uncharacterised projective simulable measuring device cannot obtain the target correlation under the aforementioned restriction even when a pair of two-qubit states are pre-shared between them. 

\begin{theo}\label{theo:cl-quant223sym}
     Sharing a pair of two-qubit states which can be used adaptively, the correlation in $T[2,2,3]$ cannot be obtained using identical projective simulable measurement which can otherwise be obtained from an identical copy of qubit non-projective measurement.
\end{theo}

\begin{proof}
 The correlation in $T[2,2,3]$ can be obtained using a pre-shared qubit singlet state and trine measurement by both Alice and Bob (See proof of theorem \ref{theo:quant223}).

We will now show the impossibility of obtaining this correlation using projective simulable measurement adaptively on a two-bipartite qubit state. Using lemma \ref{lemma:generator2} and theorem \ref{theo:classicalcorr}, we can equivalently use correlation that can be obtained using two different copies of two-level classically correlated systems. {\it W.l.o.g}, we assume that the shared classically correlated state is $\omega_{A,B}^{Cl}\otimes\mu_{A,B}^{Cl}=(\lambda,0,0,1-\lambda)\otimes (\lambda',0,0,1-\lambda')=(\lambda\lambda',0,0,0,0,\lambda(1-\lambda'),0,0,0,0,(1-\lambda)\lambda',0,0,0,0,(1-\lambda)(1-\lambda))$ where $\lambda,\lambda'\in[0,1]$.

Note, if $\lambda=0\lor 1$, then $\omega_{A,B}^{Cl}$ can be locally prepared by Alice and Bob without additional pre-shared resources. Additionally, from the theorem \ref{theo:SRG1} we know that pre-shared two-level classically correlated state cannot simulate the correlation $p(0,0)=p(1,1)=p(2,2)=0, ~p(0,1)=p(0,2)=p(1,0)=p(1,2)=p(2,0)=p(2,1)=\frac{1}{6}$. Thus, using state $\omega_{A,B}^{Cl}\otimes\mu_{A,B}^{Cl}$ where $\lambda=0\lor 1$ and $\lambda'\in[0,1]$, we cannot cannot simulate the aforementioned correlation. Using a similar argument, if $\lambda'=0\lor 1$ then using state $\omega_{A,B}^{Cl}\otimes\mu_{A,B}^{Cl}$ where  where $\lambda\in[0,1]$ we cannot cannot simulate the aforementioned correlation. Thus, $\lambda,\lambda'\in(0,1)$ for possibly obtaining the target correlation in this case.

A general three-outcome local stochastic map for Alice and Bob be 

\begin{equation}\label{eq: cl-quant223sym 1}
S_{4\rightarrow 3}:=(s_{lm})_{l=0,1,2\atop m=0,1,2,3}
\end{equation}

such that 
\begin{equation}\label{eq: cl-quant223sym 2}
\sum_{l=0}^2 s_{lm}=1~\forall~m\in \{{0,1,2,3}\}\text{ and } s_{lm}\ge 0
\end{equation}
Now, the obtained correlations from using the pre-shared state and the stochastic map are given by 
\begin{equation}
P=(S_{4\rightarrow 3}\otimes S_{4\rightarrow 3})(\omega_{A,B}^{Cl}\otimes\mu_{A,B}^{Cl})^T:= (p(a,b))_{a,b=0,1,2}
\end{equation}

Using the result, $\lambda,\lambda'\in (0,1)$ and solving the equations $p(0,0)=0$, we get that $s_{00}=s_{01}=s_{02}=s_{03}=0$. Similarly, after solving for $p(1,1)=p(2,2)=0$, we get that $s_{lm}=0~\forall~ l,m$. This solution is in contradiction with eq.(\ref{eq: cl-quant223sym 1}) and (\ref{eq: cl-quant223sym 2}).
\end{proof}

Now a natural question is whether the set of all correlations that can be generated using a pair of two-qubit pre-shared states adaptively and performing identical projective simulable measurement is a proper subset of $\mathcal{C}^{Q}_2 (3)$. The answer is negative as the correlation $p(0,0)=p(2,2)=\frac{1}{4},p(1,1)=\frac{1}{2},p(i,j\ne i)=0$ can be obtained using the former which does not belong to the latter as the mutual information of this distribution is greater than $1$.

\section{Detecting Qutrit Non-projective Measurements}\label{sec:qutrit non-proj}
In this section, we will discuss some operational tasks for detecting qutrit non-projective simulable measurement. For this purpose, here we will consider both scenarios when local measurement devices are assumed to be identical copies as well as when they could be different. In the generic case, when uncharacterised measurement devices are not restricted to being identical, the number of parameters associated with characterizing correlations obtained from projective simulable measurement for the task is quite large. Therefore, either checking if a class of target correlation is achievable by projective simulable measurement or optimizing the projective simulable bound on the figure of merit for a task becomes computationally challenging. We provide evidence using numerical optimization for the projective-simulable bounds in such a scenario. In subsection \ref{sec:S235} and \ref{sec:S335} we consider a bipartite and tripartite scenario respectively for the detection of qutrit five-outcome non-projective measurement using a figure of merit. We numerically optimise the projective simulable bound and show a violation of it using qutrit non-projective measurement. In subsection \ref{sec:S235S}, we consider a bipartite scenario for the detection of qutrit five-outcome non-projective measurement while assuming that the local measuring devices are identical. We show that some specific set of target correlation cannot be achieved using qutrit projective simulable measurement. We then show that
these correlations can be generated using qutrit non-projective measurement even in the presence of arbitrary
noise.

\subsection{Detecting $5$-outcome qutrit non-projective Measurements: $\mathbb{G}~(2,3,5)$}\label{sec:S235}

For this task $\mathbb{G}~(2,3,5)$, we consider two spatially separated parties, Alice and Bob, sharing correlated systems of local operational dimension $3$ obtained from the preparation device ${\map P}_{AB}$. They have a measurement device with $5$ outcomes, {\it i.e.}, ${\map M}_A$ and ${\map M}_B$ can yield outputs $a\in\{0,1,2,3,4\}$ and $b\in\{0,1,2,3,4\}$ respectively. We define a figure of merit for this task ${\map R}[\mathbb{G}~(2,3,5)]= [\min_{a\ne b\atop(a,b)\ne (0,4)\lor (4,0)} p(a,b)]-\sum_{i=0}^4 p(i,i)$. For this task, now we will provide a bound on the payoff that can be achieved using qutrit projective simulable measurement by both parties.
\subsubsection{Evidence for Qutrit Projective Simulable Bound}\label{sec: 235 numerical bound}
For this task, we numerically optimised the payoff for qutrit projective simulable measurement to obtain projective simulable bound. Using theorem \ref{theo:classicalcorr}, we performed this optimization over correlations generated using a shared bipartite classical system with local operational dimension $3$. The optimization is as follows:
\begin{itemize}
    \item A general classical bipartite state with local operational dimension three can be written as a row matrix: $$\omega_{A,B}^{Cl}=(\lambda_{ij})_{i,j=0,1,2}$$ such that $\sum_{i,j=0}^2 \lambda_{ij}=1$ and $\lambda_{ij}\ge 0$.
    \item A general local stochastic map $$S_{3\rightarrow 5}:=(s_{lm})_{l=0,\dots,4\atop m=0,1,2}$$ such that $\sum_{l=0}^4 s_{lm}=1$ $\forall~m\in \{{0,1,2}\}$ and $s_{lm}\ge 0$.
    \item Calculate $P:= (p(a,b))_{a,b=0,\dots,4} =(S_{3\rightarrow 5}^{A}\otimes S_{3\rightarrow 5}^{B})(\omega_{A,B}^{Cl})^T$.
    \item Maximise $p(0,1)-p(0,0)-p(1,1)-p(2,2)-p(3,3)-p(4,4)$~~~ ({\it wlog})\\such that, $p(0,1)\leq p(i,j) ~\forall~(i\ne j \text{ and } (i,j)\ne(0,4)\lor (4,0))$.
\end{itemize}
Optimizing numerically under the constraints mentioned above, we obtain the maximum value of the merit of the game $\mathbb{G}~(2,3,5)$ using classical strategies as 
\begin{equation}\label{eq:235 numeric classical}
{\map R}_{max}^{Cl}[\mathbb{G}~(2,3,5)]=1.84749\times10^{-10}\approx 0
\end{equation}
\subsubsection{Quantum Violation of ${\map R}_{max}^{Cl}[\mathbb{G}~(2,3,5)]$}
Here we will discuss a quantum strategy with a two qutrit state and quantum measurement that violates this classical/projective-simulable bound. The two-qutrit state shared between the two parties is as follows:
\begin{equation*}\label{eq:235 numeric quant state}
\begin{split}
\ket{\psi}_{AB}= \frac{1}{\sqrt{6}}(& \ket{01}_{AB}+\ket{02}_{AB}+\ket{10}_{AB}+\ket{12}_{AB}+\ket{20}_{AB}-\ket{21}_{AB})
\end{split}
\end{equation*}

Both the parties perform the POVM: 
\begin{equation*}
\begin{split}
{\map M}_A= {\map M}_B= & \{\Pi_{\alpha}=\ket{\psi_{\alpha}}\bra{\psi_{\alpha}}:\ket{\psi_{0}}=\ket{0}, \ket{\psi_{1}}=\frac{1}{\sqrt{2}}\ket{1},\ket{\psi_{2}}=\frac{1}{\sqrt{2}}\ket{2},\\
& \ket{\psi_{3}}=\frac{1}{2}(\ket{1}+\ket{2}), \ket{\psi_{4}}=\frac{1}{2}(\ket{1}-\ket{2})\}
\end{split}    
\end{equation*}
Now, we calculate the correlation obtained: $p(a,b)=\Tr \left[(\Pi_{a}^{A}\otimes \Pi_{b}^{B})|\psi\rangle\langle\psi|_{A,B}\right]$. The payoff for $\mathbb{G}~(2,3,5)$ using this strategy turns out to be
\begin{equation}
{\map R}^{Q}[\mathbb{G}~(2,3,5)]=\frac{1}{48}=0.0208>{\map R}_{max}^{Cl}[\mathbb{G}~(3,3,5)]
\end{equation}
The violation of the numerically obtained projective simulable bound on the figure of merit implies that the POVM ${\map M}_A,{\map M}_B$ are non-projective simulable measurements.
\subsection{Detecting $5$-outcome qutrit non-projective Measurements (symmetric) : $\mathbb{G}^{sym}~(2,3,5)$}\label{sec:S235S}
Now we will consider a variant of the task $\mathbb{G}~(2,3,5)$, where we assume that the uncharacterised measuring devices by both Alice and Bob are identical. Consider the set of correlations $N_{AB}$ as defined in the table \ref{tab:235symm} where $p,\epsilon\in[0,1]$. For this task, we define the winning condition as simulating the correlation of the form given in the table \ref{tab:235symm} where $\epsilon=p=1$. In other words, the target correlation $T[2,3,5]$ contains a single  correlation specified as following:

\begin{equation}\label{eq:235symtag}
    \begin{split}        &p(0,0)=p(1,1)=p(2,2)=p(3,3)=p(4,4)=0,\\
    &p(0,4)=p(4,0)=0\\    &p(0,1)=p(0,2)=p(1,0)=p(2,0)=\frac{1}{12}\\
    &p(0,3)=p(3,0)=\frac{1}{6}\\
    &p(1,2)=p(2,1)=p(3,4)=p(4,3)=\frac{1}{24}\\    &p(1,3)=p(1,4)=p(2,3)=p(2,4)=p(3,1)=\frac{1}{48}\\
    &p(3,2)=p(4,1)=p(4,2)=\frac{1}{48}
    \end{split}
\end{equation}

\begin{table}[h!]
    \centering
    \begin{tabular}{|c|c|c|c|c|c|}
    \hline
        $a\backslash b$ & $0$ & $1$ & $2$ & $3$ & $4$\\
        \hline
        $0$ & $\frac{1-\epsilon^2p}{9}$& $ \frac{2+\epsilon^2p}{36}$ & $ \frac{2+\epsilon^2p}{36}$ & $ \frac{2+\epsilon p+3\epsilon^2p}{36}$ & $ \frac{2-\epsilon p-\epsilon^2p}{36}$\\
         \hline
         $1$& $ \frac{2+\epsilon^2p}{36}$ & $ \frac{1-\epsilon^2p}{36}$ & $ \frac{2+\epsilon^2p}{72}$ & $ \frac{4+2\epsilon p-3\epsilon^2p}{144}$ & $ \frac{4-2\epsilon p+\epsilon^2p}{144}$\\
         \hline
         $2$& $ \frac{2+\epsilon^2p}{36}$ & $ \frac{2+\epsilon^2p}{72}$ & $ \frac{1-\epsilon^2p}{36}$ & $ \frac{4+2\epsilon p-3\epsilon^2p}{144}$ & $ \frac{4-2\epsilon p+\epsilon^2p}{144}$\\
        \hline
         $3$& $ \frac{2+\epsilon p+3\epsilon^2p}{36}$ & $ \frac{4+2\epsilon p-3\epsilon^2p}{144}$ & $ \frac{4+2\epsilon p-3\epsilon^2p}{144}$ & $ \frac{1+\epsilon p-2\epsilon^2p}{36}$ & $ \frac{2+\epsilon^2p}{72}$\\
         \hline
         $4$& $ \frac{2-\epsilon p-\epsilon^2p}{36}$ & $ \frac{4-2\epsilon p+\epsilon^2p}{144}$ & $ \frac{4-2\epsilon p+\epsilon^2p}{144}$ & $ \frac{2+\epsilon^2p}{72}$ & $ \frac{1-\epsilon p}{36}$\\
         \hline
    \end{tabular}
    \caption{The set of all joint probability distributions $p(a,b)$ in the set $N_{AB}$ where $p,\epsilon\in [0,1]$ for the task $\mathbb{G}^{sym}(2,3,5)$.}
    \label{tab:235symm}
\end{table}

Now we will show that this correlation cannot be obtained using two identical qutrit projective simulable measuring devices by Alice and Bob.

\subsubsection{Simulability of target correlation}

\begin{theo}\label{theo:Cl235sym}
    When using identical qutrit projective simulable measurement on a bipartite state with local operational dimension $3$, it is impossible to obtain correlation in $T[2,3,5]$.
\end{theo}

\begin{proof}
From theorem \ref{theo:classicalcorr}, correlations obtained using projective simulable measurement on a bipartite qutrit state can also be obtained using a three-level classically correlated system. We can assume that the initial shared state is classically correlated which is denoted as $\omega_{A,B}^{Cl}=(\lambda_{ij})_{i,j=0,1,2}$ such that $\sum_{i,j=0}^2 \lambda_{ij}$ and $\lambda_{ij}\ge 0$. A local stochastic map for both parties is given as follows:

\begin{equation}
S_{3\rightarrow 5}:=(s_{lm})_{l=0,\dots,4\atop m=0,1,2}
\end{equation}
such that 

\begin{equation}\label{eq:235sym stoc}
\sum_{l=0}^4 s_{lm}=1~\forall~m\in \{{0,1,2}\} \text{ and } s_{lm}\ge 0.
\end{equation}

Using the identical local stochastic map on the classical shared state above the parties obtain the correlation described by the following equation:

\begin{equation}\label{eq:235sym 1}
    P=(S_{3\rightarrow 5}^{A}\otimes S_{3\rightarrow 5}^{B})(\omega_{A,B}^{Cl})^T:= (p(a,b))_{a,b=0,\dots,4}
\end{equation} 

Now, using the expression from eq.(\ref{eq:235sym 1}) and substituting the values from eq.(\ref{eq:235symtag}), we get the following equation for probabilities of correlated outcomes:

\begin{equation}\label{eq:235sym 2}
    \begin{split}
    &p(0,0)=\lambda _{00} s_{00}^2+\lambda _{01} s_{01} s_{00}+\lambda _{02} s_{02} s_{00}+\lambda _{10} s_{01} s_{00}+\lambda _{20} s_{02} s_{00}+\lambda_{11} s_{01}^2+\lambda _{12} s_{01} s_{02}+\lambda _{21} s_{01} s_{02}+\lambda _{22} s_{02}^2=0\\
    &p(1,1)=\lambda _{00} s_{10}^2+\lambda _{01} s_{11} s_{10}+\lambda _{02} s_{12} s_{10}+\lambda _{10} s_{11} s_{10}+\lambda _{20} s_{12} s_{10}+\lambda _{11} s_{11}^2+\lambda _{12} s_{11} s_{12}+\lambda _{21} s_{11} s_{12}+\lambda _{22} s_{12}^2=0\\
    &p(2,2)=\lambda _{00} s_{20}^2+\lambda _{01} s_{21} s_{20}+\lambda _{02} s_{22} s_{20}+\lambda _{10} s_{21} s_{20}+\lambda _{20} s_{22} s_{20}+\lambda _{11} s_{21}^2+\lambda _{12} s_{21} s_{22}+\lambda _{21} s_{21} s_{22}+\lambda _{22} s_{22}^2=0\\
    &p(3,3)=\lambda _{00} s_{30}^2+\lambda _{01} s_{31} s_{30}+\lambda _{02} s_{32} s_{30}+\lambda _{10} s_{31} s_{30}+\lambda _{20} s_{32} s_{30}+\lambda _{11} s_{31}^2+\lambda _{12} s_{31} s_{32}+\lambda _{21} s_{31} s_{32}+\lambda _{22} s_{32}^2=0\\
    &p(4,4)=\lambda _{00} s_{40}^2+\lambda _{01} s_{41} s_{40}+\lambda _{02} s_{42} s_{40}+\lambda _{10} s_{41} s_{40}+\lambda _{20} s_{42} s_{40}+\lambda _{11} s_{41}^2+\lambda _{12} s_{41} s_{42}+\lambda _{21} s_{41} s_{42}+\lambda _{22} s_{42}^2=0
    \end{split}
\end{equation}

All the individual terms in the system of eq.(\ref{eq:235sym 2}) are non-negative. They must be individually zero because the probabilities are zero for correlated outcomes. Now we will show that $\lambda_{ii}= 0$ where $i\in\{0,1,2\}$. Let us assume $\lambda_{ii}\ne 0$ where $i\in\{0,1,2\}$. Then $s_{ji}=0$ where $j\in\{0,1,2,3,4\}$. This leads to a contradiction as from eq.(\ref{eq:235sym stoc}), we know that $\sum_j s_{ji}=1$.  Therefore, $\lambda_{ii}= 0$ where $i\in\{0,1,2\}$. Next substituting the values of $\lambda_{ii}$ in eq.(\ref{eq:235sym 1}) and equating them with the corresponding values from eq.(\ref{eq:235symtag}) we get that there is no solution for $\omega_{A,B}^{Cl}$ and $S_{3\rightarrow 5}$ while simultaneously satisfying the constraints given in eq.(\ref{eq:235sym stoc}).
\end{proof}
\begin{theo}\label{theo:235quant}
    The correlation in $T[2,3,5]$ can be achieved using two identical qutrit non-projective measurements.
\end{theo}
\begin{proof}
    This correlation in eq.(\ref{eq:235symtag}), can be obtained using the bipartite qutrit state $\ket{\psi}_{AB}=\frac{1}{6}(\ket{01}_{AB}+\ket{02}_{AB}+\ket{10}_{AB}+\ket{12}_{AB}+\ket{20}_{AB}-\ket{21}_{AB})$ and the measurement ${\map M}_A={\map M}_B=\{\Pi_{\alpha}=\ket{\psi_{\alpha}}\bra{\psi_{\alpha}}:\ket{\psi_{0}}=\ket{0},\ket{\psi_{1}}=\frac{1}{\sqrt{2}}\ket{1},\ket{\psi_{2}}=\frac{1}{\sqrt{2}}\ket{2}, \ket{\psi_{3}}=\frac{1}{2}(\ket{1}+\ket{2}), \ket{\psi_{4}}=\frac{1}{2}(\ket{1}-\ket{2})\}$.
\end{proof}

\subsubsection{Evidence of robustness against noise}\label{sec: 235sym robustness}
Now we will provide evidence that no correlation in $N_{AB}$, except when either $p=0$ or $\epsilon=0$, (see table \ref{tab:235symm}) can be obtained using identical qutrit projective simulable measurements by both the parties. Analogous to the proof of theorem \ref{theo:Cl235sym}, we assume that the parties share a bipartite classically correlated state $\omega_{A,B}^{Cl}=(\lambda_{ij})_{i,j=0,1,2}$ such that $\sum_{i,j=0}^2 \lambda_{ij}$ and $\lambda_{ij}\ge 0$. Both the parties locally perform some stochastic operation $S_{3\rightarrow 5}:=(s_{lm})_{l=0,\dots,4\atop m=0,1,2}$
such that $\sum_{l=0}^4 s_{lm}=1~\forall~m\in \{{0,1,2}\}$  and  $s_{lm}\ge 0$. Using the shared bipartite state and the local stochastic map will yield the correlation $P:= (p(a,b))_{a,b=0,\dots,4}=(S_{3\rightarrow 5}^{A}\otimes S_{3\rightarrow 5}^{B})(\omega_{A,B}^{Cl})^T$. Next, we equate each of the probabilities with the corresponding values from the table \ref{tab:235symm}. We solve these equations for different discreet values of $p$ and $\epsilon$ where each of these parameters takes values in the set $\{0.01,0.02, 0.03, \cdots,0.99,1\}$. We obtain that the system of equations has no solution for each of these values of $p$ and $\epsilon$. This provides numeric evidence that there is most likely no bipartite state of local operational dimension three and projective simulable measurement (identical for both parties) that yields correlations in $N_{AB}$  such that neither $p=0$ nor $\epsilon=0$.

Now, we shall show that all the correlations in $N_{AB}$ can be obtained using identical qutrit non-projective measurements by both parties. Let the parties share two qutrit state $\rho_p= p\ketbra{\psi}{\psi}+(1-p)\frac{\mathbb{I}}{3}\otimes\frac{\mathbb{I}}{3}$ where $\ket{\psi}_{AB}=\frac{1}{6}(\ket{01}_{AB}+\ket{02}_{AB}+\ket{10}_{AB}+\ket{12}_{AB}+\ket{20}_{AB}-\ket{21}_{AB})$ and $p\in (0,1]$. Both Alice and Bob perform a noisy measurement ${\map M}_A={\map M}_B=\{\Pi_{\alpha}^{\epsilon}=\lambda_{\alpha}(\mathbb{I}+\epsilon \sum_{i=1}^8v_i G_i): \epsilon\in(0,1]\}$ where $\Pi_{\alpha}= \lambda_{\alpha}(\mathbb{I}+ \sum_{i=1}^8v_i G_i)$ are the effects described in proof of theorem \ref{theo:235quant} and $\{G_i \}_{i=1}^8$ are the eight $3\times 3$ Gell Mann matrices. Using this state and measurement they can obtain all the correlations in $T[2,3,5]$ (for different values of noise $p$ and $\epsilon$).

\subsection{Detecting $5$-outcome qutrit non-projective Measurements in tripartite scenario : $\mathbb{G}~(3,3,5)$}\label{sec:S335}

We will now consider $3$ spatially separated parties sharing correlated systems of local operational dimension $3$, using a preparation device ${\map P}_{ABC}$,  for the task $\mathbb{G}~(3,3,5)$. They have a measurement device with $5$ outcomes each, {\it i.e.}, ${\map M}_A$ , ${\map M}_B$ and ${\map M}_C$ which can yield outputs $a\in\{0,1,2,3,4\}$, $b\in\{0,1,2,3,4\}$ and $c\in\{0,1,2,3,4\}$ respectively. We define a figure of merit for this task ${\map R}[\mathbb{G}~(3,3,5)]= \min_{(a,b,c)\in S} p(a,b,c)$ where set $S=\{(a,b,c):a\ne b,c=0\}\bigcup\{(a,b,c):a\ne c,b=0\}\bigcup \{(a,b,c):b\ne c,a=0\}$. We will now provide a qutrit projective simulable bound on the payoff for this task $\mathbb{G}~(3,3,5)$.
\subsubsection{Evidence for Qutrit Projective Simulable Bound}\label{sec: 335 numerical bound}

We numerically optimised the payoff for qutrit projective simulable measurement for $\mathbb{G}~(3,3,5)$. From theorem \ref{theo:classicalcorr}, $\mathcal{C}^{PQ}_3(5)=\mathcal{C}^{Cl}_3(5)$. Thus, we optimised over the classical system with local operational dimension $3$ and local stochastic maps instead for finding the qutrit projective simulable bound. The optimization problem is now as follows:

\begin{itemize}
    \item A general tripartite classical state with local operational dimension three when written as a row matrix is denoted as: $$\omega_{A_,B,C}^{Cl}=(\lambda_{ijk})_{i,j,k=0,1,2}$$ such that $\sum_{i,j,k=0}^2 \lambda_{ijk}$ and $\lambda_{ijk}\ge 0$.
    \item A general local stochastic map for the parties $A,B$ and $C$ is denoted as $$S_{3\rightarrow 5}^{\mathcal{A}}:=(s_{lm}^{\mathcal{A}})_{l=0,\dots,4\atop m=0,1,2}$$ such that $\sum_{l=0}^4 s_{lm}^{\mathcal{A}}=1$ $\forall~m\in \{{0,1,2}\}$ and $s_{lm}^{\mathcal{A}}\ge 0$ where $\mathcal{A}\in\{A,B,C\}$.
    \item Calculate $P=(S_{3\rightarrow 5}^{A}\otimes S_{3\rightarrow 5}^{B}\otimes S_{3\rightarrow 5}^{C})(\omega_{A,B,C}^{Cl})^T:= (p(a,b,c))_{a,b,c=0,\dots,4}$.
    \item Maximise $p(0,1,2)$~~~ ({\it wlog})\\such that, $p(0,1,2)\ge p(0,\beta,\gamma)$ where $\gamma \ne \beta$, $p(0,1,2)\ge p(\alpha,0,\gamma)$ where $\alpha \ne \gamma$ and $p(0,1,2)\ge p(\alpha,\beta,0)$ where $\alpha \ne \beta$.
\end{itemize}
After numerically optimizing we obtain the maximum value of the merit of the game $\mathbb{G}~(3,3,5)$ using classical strategies as $$R_{max}^{Cl}[\mathbb{G}~(3,3,5)]=0.015888$$
\subsubsection{Quantum Violation of Projective Simulable Bound}
Now we present a quantum strategy with a three qutrit state and a non-projective measurement that violates this classical/ qutrit projective-simulable bound. Let the three qutrit states shared between the three parties be: 
\begin{equation}
\begin{split}
|\psi\rangle_{A,B,C}=\frac{1}{\sqrt{6}} & (|012\rangle_{A,B,C} + |120\rangle_{A,B,C} + |201\rangle_{A,B,C} - |021\rangle_{A,B,C} -|102\rangle_{A,B,C} -|210\rangle_{A,B,C})
\end{split}
\end{equation}
All parties perform the following POVM: 
\begin{equation}
\begin{split}
{\map M}_{\mathcal{A}}:=\{& \Pi_{0}=\frac{1}{2}|0\rangle\langle 0|, \Pi_{1}=\frac{1}{2}|1\rangle\langle 1|,\Pi_{2}=\frac{1}{2}|2\rangle\langle 2|, \Pi_{3}=\frac{1}{2}(|1\rangle+|2\rangle)(\langle 1|+\langle 2|),\Pi_{4}=\frac{1}{2}(|1\rangle-|2\rangle)(\langle 1|-\langle 2|)\} 
\end{split}
\end{equation}
where $\mathcal{A}\in\{A,B,C\}$. Now the correlation they obtain is given by $p(a,b,c)=\Tr \left\{(\Pi_{a}\otimes \Pi_{b} \otimes \Pi_{c})|\psi\rangle\langle\psi|_{A,B,C}\right\}$. The merit of the game $\mathbb{G}~(3,3,5)$ using this strategy turns out to be 
\begin{equation}
R^{Q}[\mathbb{G}~(3,3,5)]=0.020833>R_{max}^{Cl}[\mathbb{G}~(3,3,5)]
\end{equation}
This result implies that the POVM ${\map M}_{\mathcal{A}}$ is non-projective. 

\section{Generalised Measurements in a class of GPT beyond Quantum Theory}\label{sec:gpt non-proj}

In this section, we consider a class of GPTs, namely boxworld \cite{Janotta11}, and show that our proposed setup can detect non-sharp simulable measurements in such theories. In recent times, this class of GPTs has been extensively studied, owing to its post-quantum nonlocal properties. One of the prime outcomes of the present study is that although such a theory allows for non-sharp simulable measurements, they can be outperformed by analogous quantum resources, providing an operational test to rule out the possibility of such theories rendering them nonphysical.    

In the boxworld $\square\equiv(\Omega,\mathcal{E})$, where $\Omega$ and $\mathcal{E}$ are the state and effect space respectively. The normalised state space for elementary systems is a regular square. The states and effects are represented by vectors in $\mathbb{R}^3$ and the probability of an effect $e\in \mathcal{E}$,  $p(e|\omega)$ is given by the Euclidean inner product. The normalised state space is the convex hull of the $4$ pure states $\{\omega_i\}_{i=0}^{3}$:
\begin{equation*}
\omega_1= \begin{pmatrix}
1 \\
0 \\
1 
\end{pmatrix}; 
\omega_2= \begin{pmatrix}
0 \\
1 \\
1 
\end{pmatrix};
\omega_3= \begin{pmatrix}
-1 \\
0 \\
1 
\end{pmatrix};
\omega_4= \begin{pmatrix}
0 \\
-1 \\
1 
\end{pmatrix}
\end{equation*}
The zero and unit effects are given by
\begin{equation*}
\mathbb{O} = \begin{pmatrix}
0 \\
0 \\
0 
\end{pmatrix};~~
u = \begin{pmatrix}
0 \\
0 \\
1 
\end{pmatrix}
\end{equation*}
The set $\mathcal{E}$ of all possible measurement effects consists of the convex hull of zero effect, unit effect, and the extremal effects $\{e_i=\frac{1}{2} \tilde{e_i}\}_{i=0}^{3}$:
\begin{equation*}
\tilde{e_1}= \begin{pmatrix}
1 \\
1 \\
1 
\end{pmatrix}; 
\tilde{e_2}= \begin{pmatrix}
-1 \\
1 \\
1 
\end{pmatrix};
\tilde{e_3}= \begin{pmatrix}
-1 \\
-1 \\
1 
\end{pmatrix};
\tilde{e_4}= \begin{pmatrix}
1 \\
-1 \\
1 
\end{pmatrix}
\end{equation*}

A measurement on an elementary system is $M:=\{f_k\}_k$ such that $f_k\in\mathcal{E}~\forall k$ and $\sum_k f_k=u$. Analogous to projective measurements in quantum theory, an elementary boxworld system has two possible sharp measurements ${E}_1:=\{e_1,e_3\}$ and ${E}_2:=\{e_2,e_4\}$. Note that, the operational dimension of an elementary boxworld system is $2$ since no set of $3$ states in $\Omega$ can be perfectly distinguished by performing a $3$-outcome measurement, while the sets of any $2$ pure states such as ${\omega_1,\omega_2}$ can be perfectly distinguished by the measurement $E_2$, {\it i.e.} $Tr~(e_2^T~\omega_1)=0$ and $Tr~(e_2^T~\omega_2)=1$.   	

The states and effects corresponding to a composition of two elementary systems can be represented by $3\times 3$ real matrices. Any bipartite composition should include all the factorised extremal states and factorised extremal effects given by: 
$$(\omega_{4i+j})_{AB}:= \omega_i \otimes \omega_j^T~~and~~(E_{4i+j})_{AB}:= e_i \otimes e_j^T.$$
A composite system also allows the possibility of states $\omega_{AB}\in\Omega_{AB}$ that cannot be prepared as a statistical mixture of the product states, {\it i.e.} $\omega_{AB}\neq\sum_ip_{ij}\omega_i\otimes\omega_j^{T}$ with $\{p_{ij}\}_{ij}$ being a probability distribution. Such states are called entangled states. Bipartite pure entangled states in the boxworld are given by
\begin{eqnarray*}
    (\omega_{17})_{AB} &=& \frac{1}{2}\left(\omega_2\otimes \omega_2^{T}-\omega_3\otimes \omega_3^{T}+\omega_3\otimes \omega_4^{T}+\omega_4\otimes \omega_3^{T}\right)\\
    (\omega_{18})_{AB} &=& \frac{1}{2}\left(\omega_1\otimes \omega_4^{T}-\omega_1\otimes \omega_1^{T}+\omega_2\otimes \omega_2^{T}+\omega_4\otimes \omega_1^{T}\right)\\
    (\omega_{19})_{AB} &=& \frac{1}{2}\left(\omega_1\otimes \omega_1^{T}-\omega_2\otimes \omega_2^{T}+\omega_2\otimes \omega_3^{T}+\omega_3\otimes \omega_2^{T}\right)\\
    (\omega_{20})_{AB} &=& \frac{1}{2}\left(\omega_1\otimes \omega_1^{T}-\omega_1\otimes \omega_4^{T}+\omega_2\otimes \omega_4^{T}+\omega_4\otimes \omega_3^{T}\right)\\
    (\omega_{21})_{AB} &=& \frac{1}{2}\left(\omega_1\otimes \omega_4^{T}-\omega_1\otimes \omega_1^{T}+\omega_2\otimes \omega_1^{T}+\omega_4\otimes \omega_2^{T}\right)\\
    (\omega_{22})_{AB} &=& \frac{1}{2}\left(\omega_1\otimes \omega_1^{T}-\omega_1\otimes \omega_2^{T}+\omega_2\otimes \omega_2^{T}+\omega_4\otimes \omega_3^{T}\right)\\
    (\omega_{23})_{AB} &=& \frac{1}{2}\left(\omega_2\otimes \omega_2^{T}-\omega_3\otimes \omega_2^{T}+\omega_3\otimes \omega_3^{T}+\omega_4\otimes \omega_1^{T}\right)\\
    (\omega_{24})_{AB} &=& \frac{1}{2}\left(\omega_1\otimes \omega_2^{T}-\omega_2\otimes \omega_2^{T}+\omega_2\otimes \omega_3^{T}+\omega_3\otimes \omega_1^{T}\right)
\end{eqnarray*}
 Entangled effects are defined similarly. Whenever such entangled states and entangled effects are invoked in a GPT, they must satisfy the basic self-consistency (SC) condition -- any valid composition of systems, states, effects, and their transformations should produce non-negative conditional probabilities. One such valid model is called the {\it PR}-model which contains the the convex hull of pure states $\{(\omega_i)_{AB}\}_{i=1}^{24}$ and the convex hull of only product effects $\{E_j\}_{j=1}^{16}$. This toy model has attracted considerable interest in the recent past \cite{Janotta14, Massar14, Safi15, Dallarno17, Saha20}. In the following, we first show that the task $\mathbb{G}(2,2,3)$ proposed in Sec. \ref{sec:S223} turns out to be useful in detecting non-projective simulable measurements in boxworld. Next, we show that in the same task $\mathbb{G}(2,2,3)$, quantum systems outperform the boxworld, which establishes the task as a testable criterion for ruling out hypothetical models of the physical world.

Similar to the classical and quantum set of correlations $\mathcal{C}^{Cl}_{2}(3)$ and $\mathcal{C}^{Q}_{2}(3)$, respectively, obtained in the bipartite setting with $3$-outcome measurements on local subsystems of dimension $2$, we define: 
\begin{itemize}
    \item the set of correlations in the {\it PR}-model as $\mathcal{C}^{PR}_{2}(3):=\left\{p(a-1,b-1)| ~p(a-1,b-1)= \right.$ $\left. Tr~(({f}_a^T)_A\otimes({g}_b)_B~\omega_{AB})~;\right.$ 
    $\left. a,b\in\{1,2,3\} \right\},$ where $\omega_{AB}\in \Omega_{AB}$ and $({f}_a)_A\in \mathcal{E}_A$, $({g}_b)_B\in \mathcal{E}_B$ with $\sum_a (f_a)_A=u_A$, $\sum_b (g_b)_B=u_B$.
    \item the set of sharp-simulable correlations in the {\it PR}-model as $\mathcal{C}^{SPR}_{2}(3):=\left\{p(a-1,b-1)|\right.$ $\left.p(a-1,b-1)=\right.$ $\left.Tr~(({f}_a^T)_A\otimes ({g}_b)_B~\omega_{AB})~;\right.$ $\left.a,b\in\{1,2,3\} \right\}$, where $\omega_{AB}\in \Omega_{AB}$ and $({f}_a)_A=$ $\sum_{i\in\{0,2\}} q_{ai}~ (e_{i})_A$, $({g}_b)_B=\sum_{j\in\{1,3\}} r_{bj}~ (e_{j})_B$ with $\sum_a ({f}_a)_A=u_A$, $\sum_b ({g}_b)_B=u_B$. $\{ q_{aj}\}_{a=1}^{3}$ and $\{ r_{bj}\}_{b=1}^{3}$ are a probability distribution for all $j$.
\end{itemize}
  
\begin{theo}
    $\mathcal{C}^{SPR}_{2}(3)\subsetneq\mathcal{C}^{PR}_{2}(3)$.
\end{theo}
\begin{proof}
We provide the proof in two steps. First we show that $\mathcal{C}^{Cl}_{2}(3)=\mathcal{C}^{SPR}_{2}(3)$. Second, to show that $\mathcal{C}^{Cl}_{2}(3)\subsetneq \mathcal{C}^{PR}_{2}(3)$. We construct a correlation that is included in the set $\mathcal{C}^{PR}_{2}(3)$ but not in $\mathcal{C}^{Cl}_{2}(3)$.

Step $1$: To prove $\mathcal{C}^{Cl}_{2}(3)=\mathcal{C}^{SPR}_{2}(3)$, let us first observe that a classical system of operational dimension $2$, {\it i.e.} a bit can always be embedded in a boxworld elementary system, with the state space $\Tilde{\Omega}:=\{Conv(\omega_1, \omega_3)\}$ and the measurement $E_1=\{e_1,e_3\}$. Here, $Conv(X,Y):=\{Z|Z=p X+(1-p) Y, 0\leq p\leq 1\}$ represents convex hull of $X$ and $Y$. Consequently, $\mathcal{C}^{Cl}_{2}(3)\subseteq\mathcal{C}^{SPR}_{2}(3)$. The set of correlations $\mathcal{C}^{Cl}_{2}(3)$ can be obtained for the bipartite state space $\Tilde{\Omega}_{AB}:=\{Conv(\omega_1\otimes\omega_1^T, \omega_1\otimes\omega_3^T, \omega_3\otimes\omega_1^T, \omega_3\otimes\omega_3^T)\}$ and the local $2$-outcome measurements $E_A=E_B=E_1$ followed by classical post-processing (stochastic maps). Second, any correlation $p^{SPR}(a-1,b-1)$ obtained from a bipartite boxworld state $\omega_{AB}=\sum_{kl} \alpha_{kl}~ \omega_{k}\otimes \omega_{l}^T$ with $\sum_{kl}\alpha_{kl}=1$ and sharp-simulable measurements, say,  $$F_A:=\{(f_{a})_A=\sum_{i\in\{0,2\}} q_{ai}~ (e_{i})_A\}_{a=1}^k$$ 
$$G_B:=\{(g_{b})_B=\sum_{j\in\{1,3\}} r_{bj}~ (e_{j})_B\}_{b=1}^k$$ 
can always be written as
\begin{eqnarray*}
    p^{SPR}(a-1,b-1) &=& Tr \left((f_{a}^T)_A\otimes (g_{b})_B~\omega_{AB} \right)\\
    &=&\sum_{ijkl} q_{ai}~r_{bj}~\alpha_{kl} Tr\left((e_i^T)_A\otimes (e_j)_B (\omega_k\otimes \omega_l^T)\right)\\
    &=&\sum_{ij} q_{ai}~r_{bj}\sum_{kl}\alpha_{kl} Tr\left((e_i^T)_A~\omega_k\right) Tr\left((e_j)^T_B~\omega_l\right)\\  
    &=&\sum_{ij} q_{ai}~r_{bj}~ Tr\left(((e_i^T)_A\otimes (e_j)_B) (\sum_{m,n}\beta_{mn}\omega_m\otimes \omega_n^T)\right)
\end{eqnarray*}
where $m,n\in\{0,2\}$ and 
$\beta_{ij}=\frac{1}{4}\sum_{kl}\alpha_{kl}Tr\left((e_i^T)_A~\omega_k\right)$ $ \cdot Tr\left((e_j)_B~\omega_l^T\right)=\frac{1}{4}(\alpha_{i,j}+\alpha_{i,j\oplus3}+\alpha_{i\oplus3,j}+\alpha_{i\oplus 3,j\oplus 3})\ge 0$ and $\sum_{ij}\beta_{ij}=1$. The last equality implies that the above correlation can also be obtained from the classical state space embedded in the boxworld state space $\Tilde{\Omega}_{AB}:=\{Conv(\omega_1\otimes\omega_1^T, \omega_1\otimes\omega_3^T, \omega_3\otimes\omega_1^T, \omega_3\otimes\omega_3^T)\}$ and local measurements simulable from measurements in computational bases (say, $E_1$ for both parties).

Step 2: Now we prove that $\mathcal{C}^{Cl}_{2}(3)\subset\mathcal{C}^{PR}_{2}(3)$.  

Let us consider the following bipartite state
\begin{equation*}
    \omega_{AB}= \frac{p_1}{2} \left(\omega_6+\omega_{16}\right) + \frac{1-p_1}{2} \left(\omega_{22}+\omega_{16}\right)
\end{equation*}
and the local measurements (see figure \ref{SQUIT:example}):
\begin{equation*}
 {F}_{A}:=
\begin{cases}
    (f_1)_A = \frac{1}{2} e_2 + \frac{1}{6} e_3\\
(f_2)_A = \frac{1}{3} e_3 + \frac{1}{3} e_4\\
(f_3)_A = \frac{1}{2} e_1 + \frac{1}{6} e_4
\end{cases};~ 
 {G}_{B}:=
\begin{cases}
    (g_1)_B = \frac{1}{2} e_2 + \frac{1}{6} e_3\\
(g_2)_B = \frac{1}{3} e_3 + \frac{1}{3} e_4\\
(g_3)_B = \frac{1}{2} e_1 + \frac{1}{6} e_4
\end{cases}
\end{equation*}
Which gives rise to the following correlation:
\begin{equation}
    T_{PR}:= \left\{p(a-1,b-1)= Tr~\left((f_a^T)_A\otimes (g_b)_B~\omega_{AB}\right)\right\} 
\end{equation}

The above correlation achieves $R^{PR}[\mathbb{G}~(2,2,3)]=0.15>R_{max}^{Cl}[\mathbb{G}~(2,2,3)]$, which implies $\mathcal{C}^{Cl}_{2}(3)\subsetneq \mathcal{C}^{PR}_{2}(3)$. 
\end{proof}

\begin{figure}[h!]
    \centering
    \includegraphics[scale=0.08]{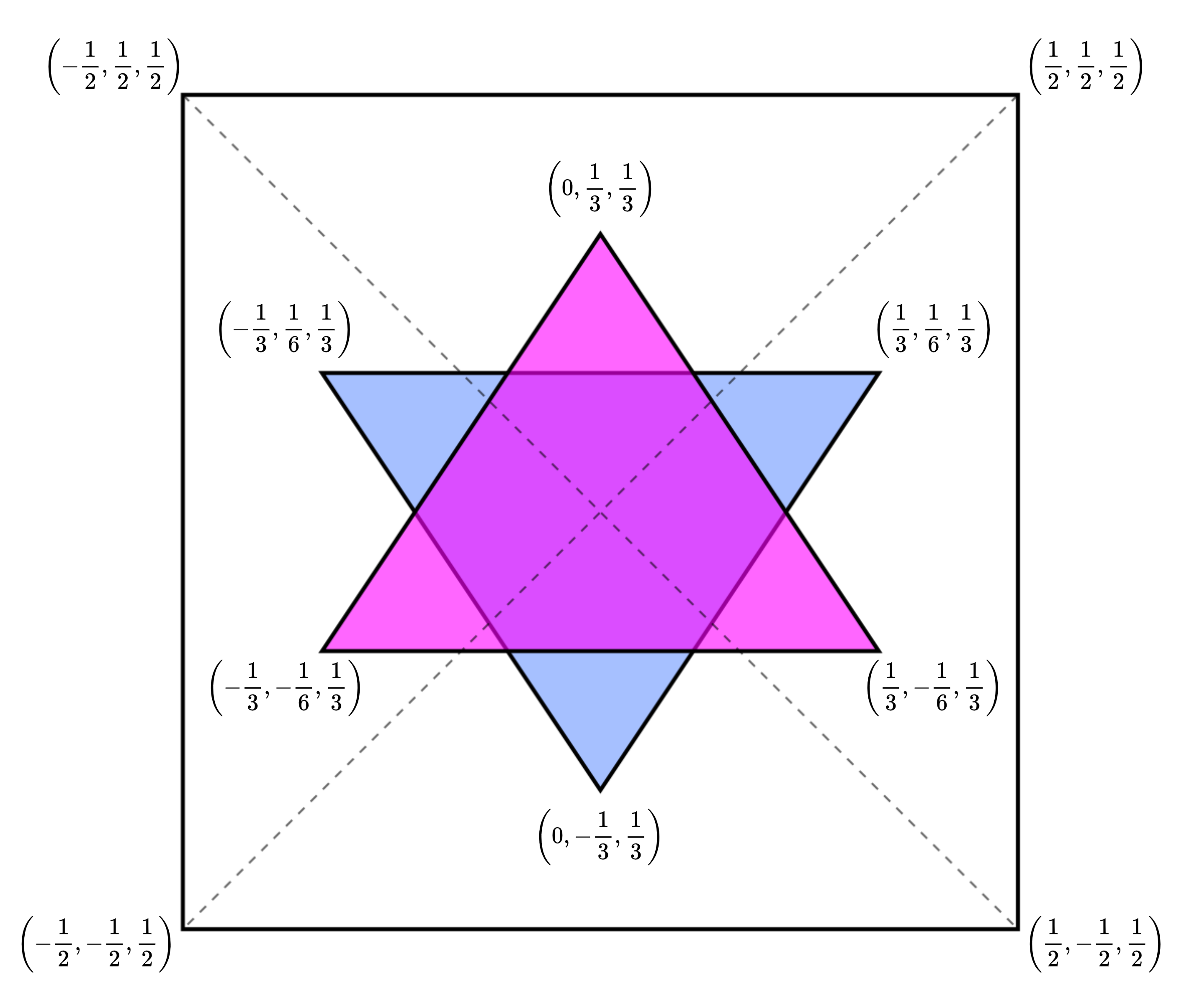}
    \caption{The local measurement strategies for obtaining $T_{PR}$: The black square denotes the local elementary effect space of the boxworld. Alice and Bob perform the measurements, each with three effects given by the vertices of the violet and blue triangles, respectively, on their subsystems of the bipartite state $\omega^{AB}$.}
    \label{SQUIT:example}
\end{figure}

\begin{theo}
    There exist a probability distribution $\tilde{P}\in \mathcal{C}^{Q}_{2}(3)$ but $\tilde{P}\notin \mathcal{C}^{S}_{2}(3)$. 
\end{theo}
\begin{proof}
    The existence of such a distribution $\tilde{P}$ is implied by the fact that $R^{PR}_{max}[\mathbb{G}~(2,2,3)]=0.1556<R_{max}^{Q}[\mathbb{G}~(2,2,3)]=\frac{1}{6}$. 
\end{proof}

\section{Summary and Discussion}\label{Discussion} 

Certifying the non-classicality of quantum measuring devices is a central problem in the emerging field of quantum technology. Besides resource-intensive processes such as device tomography, several proposals have been made and experimentally realised that only require performing random measurements from a smaller set \cite{Brunner14, Sekatski23, Vertesi10, Gomez16, Tavakoli20, Martinez23}.
In all such cases, success depends directly on the guarantee that the choices made by the experimenters are random. On the other hand, such a guarantee requires the source of randomness to be quantum. A recent work \cite{Ma2023} has shown that an alternative certification of quantum measurements is possible without seed-randomness if the experimenter only knows the upper bound of the dimension of the systems being measured. Following this line of inquiry, the present work proposes a class of operational tasks that involves simulating some particular target correlations. We consider scenarios concerning both the exact and approximate simulation of target correlations. In the latter case, a figure of merit is defined for the task that can be used to detect non-projective measurements. We first show the equivalence between correlations generated using qudit projective simulable measurements and pre-shared classical $d$-level system. Next in the bipartite scenario, we discuss the detection scheme for three and four-outcome qubit non-projective measurement. These detection schemes are robust against {\it arbitrary noise}. We further show detection schemes for five-outcome qutrit non-projective measurements in bipartite as well as tripartite scenarios. We then numerically obtain projective simulable bounds on the figure of merit when the devices are completely uncharacterised and show a violation of it using qutrit measurement on a pre-shared state. However, when the measuring devices are assumed to be identical we analytically show the impossibility of simulating target correlations with qutrit projective simulable measurement. These correlations can be used to detect qutrit non-projective measurements. Finally, we show that the tasks proposed in the present article and the earlier works could be used to {\it rule out hypothetical theories} without input randomness. For example, the well-known square bit theory (box-world) which gives rise to the PR correlation in the Bell-nonlocality setting, can be deemed unphysical if one obtains a payoff beyond a threshold value in the task $\mathbb{G}(2,2,3)$. 

This work leaves several questions open. First, analogous to the detection scheme for three and four-outcome qubit non-projective measurements discussed in subsection \ref{sec: 223 t2} and \ref{sec: 224 t2}, whether the detection schemes discussed in subsection \ref{sec: 223 t1} and \ref{sec: 224 t1} are robust against noise is not known. Second, although we show some of the correlations in $T_2[2,2,3]$ (subsection \ref{sec: 223 t2}) and $T_2[2,2,4]$ (subsection \ref{sec: 224 t2}) can be simulated using qubit non-projective measurements, an open question remains regarding the possibility of generating all correlations (with mutual information less than $1$) in this set using some qubit non-projective measurement (see region $R3$ in figure \ref{fig:Noisy223} and figure \ref{fig:Noisy224}). Third, the possibility of analytically obtaining a nontrivial upper bound on the projective simulable bound discussed in subsection \ref{sec: 235 numerical bound} and \ref{sec: 335 numerical bound} remains open. This would provide sufficient proof for the detection schemes discussed in these subsections. Fourth, in subsection \ref{sec: 235sym robustness} by varying the parameters in the target correlation $N_{AB}$ we show the impossibility of generating large (but finite) correlations in this set using qutrit projective simulable measurements. It remains open whether all the correlations in this set cannot be obtained using qutrit projective simulable measurements. This would prove that the detection scheme is robust against arbitrary white noise. Finally, an open question regarding the last part of the article is whether the task $\mathbb{G}(2,2,3)$ can be used to rule out a class of GPTs called {\it polygon models}, which has been shown to mimic quantum statistics for the number of vertices $n\rightarrow \infty$. \cite{Janotta11}.

\section{Acknowledgement}
 S. R. and P. H. acknowledge partial support by the Foundation for Polish Science -- IRAP project, ICTQT, contract no. MAB/2018/5, which is carried out within the International Research Agendas Programme of the Foundation for Polish Science co-financed by the European Union from the funds of the Smart Growth Operational Programme, axis IV: Increasing the research potential (Measure 4.3). SSB acknowledges funding by the Spanish MICIN (project PID2022-141283NB-I00) with the support of FEDER funds, and by the Ministry for Digital Transformation and of Civil Service of the Spanish Government through the QUANTUM ENIA project call- Quantum Spain project, and by the European Union through
the Recovery, Transformation and Resilience Plan - NextGeneration EU within the framework of the Digital Spain 2026 Agenda.

\appendix
\section{Proofs}
\subsection{Proof of Theorem \ref{theo:classicalcorr}}\label{app:classicalcorr}

\begin{proof}
   Here we will consider the proof when the number of parties $n=2$ in the task $\mathbb{G}~(n,d,k)$. The proof for a general case with arbitrary $n$ will be a simple extension of the proof which we discuss here.  In the following we first show that given operational dimension $d$ of the local system, for any joint outcome probabilities $\{p^{PQ}(a,b)\}_{a,b=1}^k\in \mathcal{C}^{PQ}_d (k)$ obtained from some arbitrary quantum state and two local projective-simulable measurements there is a correlation $\{q^{cl}(r,s)\}_{r,s=1}^d$ obtained using a pseudo-quantum density matrix (a classical state) and local computational basis measurements that generates the same correlation as $\{p^{PQ}(a,b)\}_{a,b=1}^k$ after the post-processing of outcomes. In other words,  $\mathcal{C}^{PQ}_d (k)\subseteq \mathcal{C}^{Cl}_d (k)$.

   First, note that the statistics $\{p^{PQ}(a,b)\}$ generated from some local projective-simulable measurements $M^A:=\{e_a= \sum_{m=1}^{d} p_{am}|\psi_m\rangle\langle\psi_m|\}_{a=1}^k$ and $M^B:=\{f_b= \sum_{n=1}^{d} q_{bn}|\phi_n\rangle\langle\phi_n|\}_{b=1}^k$, respectively acting on a arbitrary bipartite quantum state $\rho\in\mathbb{C}^d\otimes\mathbb{C}^d$ can be written as
   \begin{eqnarray*}
       p^{PQ}(a,b) &=& Tr\left(\rho ~e_a\otimes f_b\right)\\
       &=& Tr\left(\rho \sum_{m} p_{am}|\psi_m\rangle\langle\psi_m|\otimes \sum_{n} q_{bn}|\phi_n\rangle\langle\phi_n|\right)\\
       &=& \sum_{m,n} p_{am} q_{bn} Tr\left(\rho ~|\psi_m\rangle\langle\psi_m|\otimes |\phi_n\rangle\langle\phi_n|\right)
   \end{eqnarray*}
Here $\{\ket{\psi_m}\}_{m=1}^d$ and $\{\ket{\phi_n}\}_{n=1}^d$ form an orthonormal basis of $\mathbb{C}^d$ and $\{p_{am}\}_{m=1}^d$ and $\{q_{bn}\}_{n=1}^d$ are a valid probability distribution $\forall ~m,n$ respectively. Now any arbitrary bipartite quantum state $\rho\in \mathcal{D}(\mathbb{C}^d\otimes\mathbb{C}^d)$ can be written as following:
$$\rho=\sum_{ijkl=1}^{d} \alpha_{ijkl}~|\psi_i\phi_j\rangle\langle\psi_k\phi_l|.$$
This leads us to the following expression for $p^{PQ}(a,b)$:

\begin{eqnarray*}
       p^{PQ}(a,b) &=& \sum_{m,n} p_{am} q_{bn} Tr\left(\rho |\psi_m\phi_n\rangle\langle\psi_m\phi_n|\right)\\
       &=& \sum_{m,n} p_{am} q_{bn} \sum_{ijkl=1}^{d} \alpha_{ijkl}~Tr\left(|\psi_i\phi_j\rangle\langle\psi_k\phi_l|\cdot|\psi_m\phi_n\rangle\langle\psi_m\phi_n|\right)\\
       &=& \sum_{m,n} p_{am} q_{bn} \sum_{ijkl=1}^{d} \alpha_{ijkl}~\delta_{mk}\delta_{ln} Tr\left(|\psi_i\phi_j\rangle\langle\psi_m\phi_n|\right)\\
       &=& \sum_{m,n} p_{am} q_{bn} \sum_{ij=1}^{d} \alpha_{ijmn}~ |\langle\psi_i\phi_j|\psi_m\phi_n\rangle|^2\\
       &=&\sum_{m,n} p_{am} q_{bn} \sum_{ij=1}^{d} \alpha_{ijmn} ~\delta_{im}\delta_{jn}\\
       &=& \sum_{m,n} p_{am} q_{bn} \alpha_{mnmn}\\
       &=& \sum_{m,n} p_{am} q_{bn} \sum_{ij=1}^{d} \alpha_{ijij}~ Tr\left(|\psi_i\phi_j\rangle\langle\psi_i\phi_j|\cdot|\psi_m\phi_n\rangle\langle\psi_m\phi_n|\right)\\
       &=& \sum_{m,n} p_{am} q_{bn} \sum_{ij=1}^{d} \alpha_{ijij}~ Tr\left(|ij\rangle\langle ij|\cdot|mn\rangle\langle mn|\right)\\
       &=& \sum_{m,n} p_{am} q_{bn} Tr\left(\tilde{\rho}~|mn\rangle\langle mn|\right)\\
       &=&  \sum_{m,n} p_{am} q_{bn}~ q^{cl}(m,n)\\
       &=& p^{cl}(a,b)
   \end{eqnarray*}
      
where $\tilde{\rho}=\sum_{ij=1}^{d} \alpha_{ijij}~|ij\rangle\langle ij|$ is a pseudo-quantum state (classical probability distribution) and $\{|i\rangle\}_{i=1}^d$ are the computational bases for the two parties. Now, it is easy to see that the converse $\mathcal{C}^{Cl}_d (k)\subseteq \mathcal{C}^{PQ}_d (k)$ is true as any correlation $\{p^{cl}(a,b)\}_{a,b=1}^k$ obtained from some classical state with local operational dimension $d$, {\it i.e.} $\{p^{cl}(r,s)\}_{r,s=1}^d$, and therefore can be trivially obtained from a pseudo-quantum state and a measurement in computational bases followed by local post-processing of outcomes for the two parties. Moreover, the generalisation of the above proof for an arbitrary number of parties $n$ can be simply done by assuming a local projective simulable measurement for each party along with a multipartite shared state among them. The steps of the proof will thus be exactly similar to the one shown above. This completes the proof.
\end{proof}

\bibliographystyle{ieeetr}

\begin{thebibliography}{10}

\bibitem{Ivanovic87}
I.~Ivanovic, ``How to differentiate between non-orthogonal states,'' {\em Physics Letters A}, vol.~123, no.~6, pp.~257--259, 1987.

\bibitem{Peres88}
A.~Peres, ``How to differentiate between non-orthogonal states,'' {\em Physics Letters A}, vol.~128, no.~1, p.~19, 1988.

\bibitem{DiMario22}
M.~T. DiMario and F.~E. Becerra, ``Demonstration of optimal non-projective measurement of binary coherent states with photon counting,'' {\em npj Quantum Information}, vol.~8, p.~84, July 2022.

\bibitem{QuantumStateDisriminationRev}
J.~Bae and L.-C. Kwek, ``Quantum state discrimination and its applications,'' {\em Journal of Physics A: Mathematical and Theoretical}, vol.~48, p.~083001, Jan 2015.

\bibitem{Shang2018}
J.~Shang, A.~Asadian, H.~Zhu, and O.~G\"uhne, ``Enhanced entanglement criterion via symmetric informationally complete measurements,'' {\em Phys. Rev. A}, vol.~98, p.~022309, Aug 2018.

\bibitem{Derka1998}
R.~Derka, V.~Buz\ifmmode~\breve{}\else \u{}\fi{}ek, and A.~K. Ekert, ``Universal algorithm for optimal estimation of quantum states from finite ensembles via realizable generalized measurement,'' {\em Phys. Rev. Lett.}, vol.~80, pp.~1571--1575, Feb 1998.

\bibitem{Renes2004}
J.~M. Renes, R.~Blume-Kohout, A.~J. Scott, and C.~M. Caves, ``Symmetric informationally complete quantum measurements,'' {\em Journal of Mathematical Physics}, vol.~45, no.~6, pp.~2171--2180, 2004.

\bibitem{OptTomography}
J.~{Haah}, A.~W. {Harrow}, Z.~{Ji}, X.~{Wu}, and N.~{Yu}, ``Sample-optimal tomography of quantum states,'' {\em IEEE Transactions on Information Theory}, vol.~63, no.~9, pp.~5628--5641, 2017.

\bibitem{singleSettingTom}
R.~Stricker, M.~Meth, L.~Postler, C.~Edmunds, C.~Ferrie, R.~Blatt, P.~Schindler, T.~Monz, R.~Kueng, and M.~Ringbauer, ``Experimental single-setting quantum state tomography,'' {\em PRX Quantum}, vol.~3, p.~040310, Oct 2022.

\bibitem{povmSHADOWGhune}
H.~C. Nguyen, J.~L. B\"onsel, J.~Steinberg, and O.~G\"uhne, ``Optimizing shadow tomography with generalized measurements,'' {\em Phys. Rev. Lett.}, vol.~129, p.~220502, Nov 2022.

\bibitem{Ragy2016}
S.~{Ragy}, M.~{Jarzyna}, and R.~{Demkowicz-Dobrza{\'n}ski}, ``{Compatibility in multiparameter quantum metrology},'' {\em \pra}, vol.~94, p.~052108, Nov. 2016.

\bibitem{Szcyykulska2016}
M.~Szczykulska, T.~Baumgratz, and A.~Datta, ``Multi-parameter quantum metrology,'' {\em Advances in Physics: X}, vol.~1, no.~4, pp.~621--639, 2016.

\bibitem{rafalMETRO}
F.~Albarelli and R.~Demkowicz-Dobrza\ifmmode~\acute{n}\else \'{n}\fi{}ski, ``Probe incompatibility in multiparameter noisy quantum metrology,'' {\em Phys. Rev. X}, vol.~12, p.~011039, Mar 2022.

\bibitem{Acin16}
A.~Ac\'{\i}n, S.~Pironio, T.~V\'ertesi, and P.~Wittek, ``Optimal randomness certification from one entangled bit,'' {\em Phys. Rev. A}, vol.~93, p.~040102, Apr 2016.

\bibitem{optimalCRYPTO}
C.~A. Fuchs, N.~Gisin, R.~B. Griffiths, C.-S. Niu, and A.~Peres, ``Optimal eavesdropping in quantum cryptography. i. information bound and optimal strategy,'' {\em Phys. Rev. A}, vol.~56, pp.~1163--1172, Aug 1997.

\bibitem{Ishizaka2008}
S.~Ishizaka and T.~Hiroshima, ``Asymptotic teleportation scheme as a universal programmable quantum processor,'' {\em Phys. Rev. Lett.}, vol.~101, p.~240501, Dec 2008.

\bibitem{Studzinski2017port}
M.~Studzi{\'{n}}ski, S.~Strelchuk, M.~Mozrzymas, and M.~Horodecki, ``Port-based teleportation in arbitrary dimension,'' {\em Scientific Reports}, vol.~7, p.~10871, Sep 2017.

\bibitem{Mozrzymas2018optimal}
M.~Mozrzymas, M.~Studzi{\'{n}}ski, S.~Strelchuk, and M.~Horodecki, ``Optimal port-based teleportation,'' {\em New Journal of Physics}, vol.~20, p.~053006, May 2018.

\bibitem{algREVIEW}
A.~M. Childs and W.~van Dam, ``Quantum algorithms for algebraic problems,'' {\em Rev. Mod. Phys.}, vol.~82, pp.~1--52, Jan 2010.

\bibitem{Bacon2006}
D.~Bacon, A.~M. Childs, and W.~v. Dam, ``Optimal measurements for the dihedral hidden subgroup problem,'' {\em Chicago Journal of Theoretical Computer Science}, vol.~2006, October 2006.

\bibitem{HSP}
P.~{Sen}, ``Random measurement bases, quantum state distinction and applications to the hidden subgroup problem,'' in {\em 21st Annual IEEE Conference on Computational Complexity (CCC'06)}, pp.~14 pp.--287, 2006.

\bibitem{Barret2002}
J.~Barrett, ``Nonsequential positive-operator-valued measurements on entangled mixed states do not always violate a bell inequality,'' {\em Phys. Rev. A}, vol.~65, p.~042302, Mar 2002.

\bibitem{Vertesi2010}
T.~V\'ertesi and E.~Bene, ``Two-qubit bell inequality for which positive operator-valued measurements are relevant,'' {\em Phys. Rev. A}, vol.~82, p.~062115, Dec 2010.

\bibitem{Lundeen09}
J.~S. Lundeen, A.~Feito, H.~Coldenstrodt-Ronge, K.~L. Pregnell, C.~Silberhorn, T.~C. Ralph, J.~Eisert, M.~B. Plenio, and I.~A. Walmsley, ``Tomography of quantum detectors,'' {\em Nature Physics}, vol.~5, pp.~27--30, Jan. 2009.

\bibitem{Brunner14}
N.~Brunner, D.~Cavalcanti, S.~Pironio, V.~Scarani, and S.~Wehner, ``Bell nonlocality,'' {\em Rev. Mod. Phys.}, vol.~86, pp.~419--478, Apr 2014.

\bibitem{Sekatski23}
P.~Sekatski, J.-D. Bancal, M.~Ioannou, M.~Afzelius, and N.~Brunner, ``Toward the device-independent certification of a quantum memory,'' {\em Phys. Rev. Lett.}, vol.~131, p.~170802, Oct 2023.

\bibitem{Barrett11}
J.~Barrett and N.~Gisin, ``How much measurement independence is needed to demonstrate nonlocality?,'' {\em Phys. Rev. Lett.}, vol.~106, p.~100406, Mar 2011.

\bibitem{Putz14}
G.~P\"utz, D.~Rosset, T.~J. Barnea, Y.-C. Liang, and N.~Gisin, ``Arbitrarily small amount of measurement independence is sufficient to manifest quantum nonlocality,'' {\em Phys. Rev. Lett.}, vol.~113, p.~190402, Nov 2014.

\bibitem{Vertesi10}
T.~V\'ertesi and E.~Bene, ``Two-qubit bell inequality for which positive operator-valued measurements are relevant,'' {\em Phys. Rev. A}, vol.~82, p.~062115, Dec 2010.

\bibitem{Gomez16}
E.~S. G\'omez, S.~G\'omez, P.~Gonz\'alez, G.~Ca\~nas, J.~F. Barra, A.~Delgado, G.~B. Xavier, A.~Cabello, M.~Kleinmann, T.~V\'ertesi, and G.~Lima, ``Device-independent certification of a nonprojective qubit measurement,'' {\em Phys. Rev. Lett.}, vol.~117, p.~260401, Dec 2016.

\bibitem{Tavakoli20}
A.~Tavakoli, M.~Smania, T.~Vértesi, N.~Brunner, and M.~Bourennane, ``Self-testing nonprojective quantum measurements in prepare-and-measure experiments,'' {\em Science Advances}, vol.~6, no.~16, p.~eaaw6664, 2020.

\bibitem{Martinez23}
D.~Mart{\'\i}nez, E.~S. G{\'o}mez, J.~Cari{\~n}e, L.~Pereira, A.~Delgado, S.~P. Walborn, A.~Tavakoli, and G.~Lima, ``Certification of a non-projective qudit measurement using multiport beamsplitters,'' {\em Nature Physics}, vol.~19, pp.~190--195, Feb. 2023.

\bibitem{Renou19}
M.-O. Renou, E.~B\"aumer, S.~Boreiri, N.~Brunner, N.~Gisin, and S.~Beigi, ``Genuine quantum nonlocality in the triangle network,'' {\em Phys. Rev. Lett.}, vol.~123, p.~140401, Sep 2019.

\bibitem{Guha2021}
T.~Guha, M.~Alimuddin, S.~Rout, A.~Mukherjee, S.~S. Bhattacharya, and M.~Banik, ``Quantum {A}dvantage for {S}hared {R}andomness {G}eneration,'' {\em {Quantum}}, vol.~5, p.~569, Oct. 2021.

\bibitem{Ma2023}
Z.~Ma, M.~Rambach, K.~Goswami, S.~S. Bhattacharya, M.~Banik, and J.~Romero, ``Randomness-free test of nonclassicality: A proof of concept,'' {\em Phys. Rev. Lett.}, vol.~131, p.~130201, Sep 2023.

\bibitem{Janotta11}
P.~Janotta, C.~Gogolin, J.~Barrett, and N.~Brunner, ``Limits on nonlocal correlations from the structure of the local state space,'' {\em New Journal of Physics}, vol.~13, p.~063024, Jun 2011.

\bibitem{Scandolo21}
C.~M. Scandolo, R.~Salazar, J.~K. Korbicz, and P.~Horodecki, ``Universal structure of objective states in all fundamental causal theories,'' {\em Phys. Rev. Res.}, vol.~3, p.~033148, Aug 2021.

\bibitem{Ambainis08}
A.~Ambainis, D.~Leung, L.~Mancinska, and M.~Ozols, ``Quantum random access codes with shared randomness,'' {\em arXiv preprint 0810.2937}, 2009.

\bibitem{Nielsen2012}
M.~A. Nielsen and I.~L. Chuang, {\em Quantum Computation and Quantum Information: 10th Anniversary Edition}.
\newblock Cambridge University Press, June 2012.

\bibitem{Oszmaniec2017}
M.~Oszmaniec, L.~Guerini, P.~Wittek, and A.~Ac\'{\i}n, ``Simulating positive-operator-valued measures with projective measurements,'' {\em Phys. Rev. Lett.}, vol.~119, p.~190501, Nov 2017.

\bibitem{Singal2022}
T.~Singal, F.~B. Maciejewski, and M.~Oszmaniec, ``Implementation of quantum measurements using classical resources and only a single ancillary qubit,'' {\em npj Quantum Information}, vol.~8, July 2022.

\bibitem{Chiribella11}
G.~Chiribella, G.~M. D'Ariano, and P.~Perinotti, ``Informational derivation of quantum theory,'' {\em Phys. Rev. A}, vol.~84, p.~012311, Jul 2011.

\bibitem{Dallarno17}
M.~Dall'Arno, S.~Brandsen, A.~Tosini, F.~Buscemi, and V.~Vedral, ``No-hypersignaling principle,'' {\em Phys. Rev. Lett.}, vol.~119, p.~020401, Jul 2017.

\bibitem{Frenkel15}
P.~E. Frenkel and M.~Weiner, ``Classical information storage in an n-level quantum system,'' {\em Communications in Mathematical Physics}, vol.~340, pp.~563--574, Dec. 2015.

\bibitem{Kleinmann17}
M.~Kleinmann, T.~V\'ertesi, and A.~Cabello, ``Proposed experiment to test fundamentally binary theories,'' {\em Phys. Rev. A}, vol.~96, p.~032104, Sep 2017.

\bibitem{wilde2011classical}
M.~M. Wilde, ``From classical to quantum shannon theory,'' {\em arXiv preprint arXiv:1106.1445}, 2011.

\bibitem{Janotta14}
P.~Janotta and H.~Hinrichsen, ``Generalized probability theories: what determines the structure of quantum theory?,'' {\em Journal of Physics A: Mathematical and Theoretical}, vol.~47, p.~323001, Jul 2014.

\bibitem{Massar14}
S.~Massar and M.~K. Patra, ``Information and communication in polygon theories,'' {\em Phys. Rev. A}, vol.~89, p.~052124, May 2014.

\bibitem{Safi15}
S.~W. Al-Safi and J.~Richens, ``Reversibility and the structure of the local state space,'' {\em New Journal of Physics}, vol.~17, p.~123001, Dec 2015.

\bibitem{Saha20}
S.~Saha, S.~S. Bhattacharya, T.~Guha, S.~Halder, and M.~Banik, ``Advantage of quantum theory over nonclassical models of communication,'' {\em Annalen der Physik}, vol.~532, no.~12, p.~2000334, 2020.

\end{thebibliography}

\end{document}